\documentclass[runningheads]{llncs}
\usepackage[T1]{fontenc}

\usepackage{times}  
\usepackage{helvet} 
\usepackage{courier}  
\usepackage[hyphens]{url} 
\usepackage{graphicx} 
\usepackage{caption} 
\usepackage{latexsym}
\usepackage{amssymb}

\usepackage{amsmath}
\usepackage{amsthm}
\usepackage{booktabs}
\usepackage{enumitem}
\usepackage{graphicx}
\usepackage{tikz}
\usepackage{tikz-network}
\usepackage{tabularx}
\usepackage{color}
\usepackage{xcolor}
\usepackage{todonotes}
\usepackage{algorithm, algpseudocode}
\usepackage{esvect}
\usepackage{siunitx}
\usepackage{booktabs}
\usepackage{multirow}
\usepackage{dblfloatfix}
\usepackage{pdflscape}
\usepackage[dvipsnames]{xcolor}

\usepackage{newfloat}
\usepackage{listings}
\usepackage{authblk}

\newcommand{\tree}{\mathcal{T}}
\newcommand{\nodes}{\mathcal{N}}
\newcommand{\edges}{\mathcal{E}}
\newcommand{\items}{\mathcal{G}}

\newcommand{\children}{\mathcal{C}}
\newcommand{\parent}{\mathcal{P}}
\newcommand{\leaves}{\mathcal{L}}
\newcommand{\internal}{\mathcal{I}}
\newcommand{\ancestors}{Anc}

\newcommand{\mallocations}{\mathrm{\Pi}}
\newcommand{\allocations}{\mathcal{A}}

\newenvironment{proofsketch}{
  \noindent \textit{Proof sketch.}
}{}

\begin{document}
\title{Multilevel Fair Allocation under Additive Preferences}
\author{
Maxime Lucet
\and
Nawal Benabbou \and
Aur{\'e}lie Beynier \and
Nicolas Maudet\\
\email
\{firstname.lastname\}@lip6.fr
}
% \author{
% Maxime Lucet\orcidID{0009-0003-0774-5443}
% \and
% Nawal Benabbou\orcidID{0000-0002-4589-4162}\and
% Aurélie Beynier\orcidID{0000-0002-3812-6755}\and
% Nicolas Maudet\orcidID{0000-0002-4232-069X}\\
% \emails
% \{firstname.lastname\}@lip6.fr
% }
%
\authorrunning{M. Lucet et al.}

% \author{Paper ID 8}
% \institute{LIP6, CNRS, Sorbonne Université}
\institute{LIP6, CNRS, Sorbonne Université, F-75005 Paris, France}
\maketitle              % typeset the header of the contribution
\begin{abstract}
We study multilevel fair resource allocation with tree-structured hierarchical relations among agents. At each level, the problem can be viewed locally as allocating an agent's bundle to its children, the overall allocation being a trace of this process iterated down to the leaves.
% with each internal node potentially having its own local allocation mechanism. 
Assuming that internal nodes' utilities are the utilitarian welfare of their children, and the leaves have classical additive utilities over items, we first propose multilevel adaptations of usual envy-based fairness notions (e.g., WEF1). We present three adaptations and show that the choice among them is not neutral. We prove that, under identical preferences, the three adapted envy-based notions coincide, and that the Multilevel extension of Weighted Round Robin \cite{Chakraborty_et_al_2020} (MWRR) guarantees them. We then show that under general preferences, MWRR  may guarantee some notions while failing others. Finally, through experiments, we show that MWRR may still perform well even for adaptations it does not formally guarantee.

\end{abstract}

\section{Introduction}

Fair allocation has traditionally focused on designing algorithms to allocate fairly resources between individual agents \cite{Caragiannis_et_al2019,Lipton_et_al2004}, and was then later extended to groups \cite{Benabbou_et_al_2019,Chakraborty_et_al_2020,Kyropoulou_2019,Aleksandrov_Walsh_2018,GrossHumbert_et_al2023}. However, in many real-world settings, agents/groups actually belong to higher level entities, ultimately forming a hierarchical structure. Such hierarchy can be naturally modelled by a directed tree, where the root is responsible for allocating the bundle of items it receives to its children. This new setting, called multilevel fair allocation, was recently formalized in \cite{Lucet_et_al2025}. 

Multilevel fair allocation problems can arise in multiple real-world situations. For instance, consider a food charity operating in a country with multiple administrative levels: regional branches oversee departmental units, which in turn manage city-level distribution centres with potentially heterogeneous supply needs (see Fig.~\ref{fig:initial_tree} for  an illustration).
%\color{black} Hence, the food charity operates as a hierarchical structure (partly) represented in Fig.~\ref{fig:initial_tree}. 
In such a setting, fairness should be enforced locally at each level of the hierarchy (e.g., among departments within a region or among cities within a department), rather than across levels, as entities at different levels are not directly comparable. Note that many allocation problems across territories, or within a hierarchical structure, can be covered by this multilevel setting.

Interestingly, this setting can be useful even in contexts where no explicit hierarchy exists. For instance, when assigning medical resources to patients characterized by severity, age, and risk level, one may first prioritize fairness across severity (e.g., giving higher priority to the most severe cases), then within severe cases place greater weight on high-risk patients, while for less severe cases one may instead balance allocation across age groups to avoid age bias (see Fig. 2 for an illustration). Such implicit hierarchy may also arise in affirmative action schemes.

\vspace{-0.55cm}
\begin{figure}[h!]
\begin{minipage}{0.48\linewidth}
\centering
\resizebox{\linewidth}{!}{
    \begin{tikzpicture}[
        level distance=10mm,
        every node/.style={rectangle, rounded corners=4pt, draw, 
                           inner sep=6pt, font=\normalsize},
        level 1/.style={sibling distance=60mm},
        level 2/.style={sibling distance=30mm},
        edge from parent/.style={draw, -stealth}
    ]
    
    \node {1: Food charity}
        child { node {2: RegionA}
            child { node {4: CityA} }
            child { node {5: CityB} }
        }
        child { node {3: RegionB}
            child { node {6: CityC} }
            child { node {7: CityD} }
        };
    
    \end{tikzpicture}
    }
    \caption{\small Hierarchical structure of a food charity.}
    \label{fig:initial_tree}
\end{minipage}
\hfill
\begin{minipage}{0.48\linewidth}
\centering
\resizebox{\linewidth}{!}{
\begin{tikzpicture}[
        level distance=10mm,
        every node/.style={rectangle, rounded corners=4pt, draw, 
                           inner sep=6pt, font=\normalsize},
        level 1/.style={sibling distance=60mm},
        level 2/.style={sibling distance=25mm},
        edge from parent/.style={draw, -stealth}
    ]
    
    \node {Hospital}
        child { node {Low-severity}
            child { node {$<18$ years} }
            child { node {$18-45$ years} }
            child { node {$>45$ years} }
        }
        child { node {High-severity}
            child { node {High-risk} }
            child { node {Low-risk} }
        };
    
    \end{tikzpicture}
    }
    \caption{\small Example of implicit hierarchy.}
    \label{fig:implicit_hierarchy}
\end{minipage}
\end{figure}
\vspace{-0.7cm}

In this paper, we assume that intermediary nodes act solely as representative bodies of their constituents (as in \cite{Lucet_et_al2025}), and therefore receive items only to allocate them to their children. 
These entities derive utility from the satisfaction of their constituents\footnote{Other assumptions could have been made, e.g., they could be oblivious of how the items are allocated to their children, or only care about the properties of the allocation to their children.}, and we adopt the utilitarian social welfare as an aggregation rule (following \cite{Lucet_et_al2025}), which is a standard modelling choice in collective decision-making.  However, while \cite{Lucet_et_al2025} assume matroid-rank utilities at the leaves, we assume additive preferences. The additive utility model is the most studied preference model in the literature and is more expressive over singletons than matroid-rank utilities. Coming back to our illustrative example, this assumption is natural in the context of food charities (see e.g., \cite{Aggarwal_et_al2024}).

Moreover, we further depart from \cite{Lucet_et_al2025} by considering different notions of fairness. More precisely, they consider welfare-based fairness notions (e.g., Lorenz dominance) ensuring that, under matroid-rank utilities, a utilitarian-optimal solution is also fair. In contrast, we consider envy-based fairness notions. 
A central notion of fairness in classical fair division is \emph{envy-freeness} (EF), together with its relaxations: \emph{envy-freeness up to one good} (EF1) \cite{Caragiannis_et_al2019} and \emph{envy-freeness up to any good} (EFX). Recent works have also examined \emph{weighted} variants of these criteria \cite{Chakraborty_et_al_2020,Chakraborty_et_al2024}, namely WEF1 and WEFX. The general idea behind these notions is that no agent should envy the bundle received by another agent, after accounting for the weight differences. In the multilevel setting, envy-based notions require some adaptations. 
 Indeed, an internal node does not observe how an alternative bundle would be allocated among its children, and thus cannot compute the exact utilitarian welfare it would derive from it, but only an estimate. 
In this paper, we consider three possible estimation functions (pessimistic, agnostic, and optimistic), and we investigate whether a multilevel version of Weighted Round Robin (WRR) \cite{Chakraborty_et_al_2020} can guarantee some of our fairness notions together with completeness (i.e., all items are allocated). In monolevel settings, the Round Robin \cite{Caragiannis_et_al2019} and WRR are indeed simple algorithms achieving EF1 and WEF1, respectively.

\vspace{0.1cm}
\noindent \textbf{Related work.} \emph{Group fairness} has become a highly active topic in fair division \cite{Benabbou_et_al_2019,Chakraborty_et_al_2020,Kyropoulou_2019,Aleksandrov_Walsh_2018,GrossHumbert_et_al2023}, and was recently extended to multilevel settings by \cite{Lucet_et_al2025}. While they focus on matroid-rank valuations and fairness notions unrelated to envy, we consider additive valuations and envy-based fairness notions in multilevel settings. 
However, some previous works can be seen as bilevel settings. \cite{Scarlett_et_al2023} study individual and group envy-freeness when groups’ valuations depend on those of their members, and propose algorithms guaranteeing fairness at both levels. Our work extends this setting, as well as some of their algorithmic ideas. Closely related, \cite{Bu_et_al2024} reconcile group and individual perspectives, allowing group preferences that do not necessarily aggregate those of their members (unlike our assumption). However, their model is limited to two levels, whereas ours allows an arbitrary number of levels. Finally, \cite{Aggarwal_et_al2024} study a bilevel allocation problem motivated by food charities, but focus on auction mechanisms, in contrast with our setting.
Our multilevel model is also related to the multilevel apportionment framework of \cite{Schmidt_et_al2025}, which uses a similar tree-based hierarchy. However, their work focuses specifically on apportionment, a very specific subfield of resource allocation.
Finally, the multilevel fairness notions proposed in this paper induce some visibility constraints, which is reminiscent of work on local fairness \cite{Abebe_et_al_2017,Beynier_et_al2019,BREDERECK_2022}; for instance, in our model, only labs affiliated to the same department can compare their situations.

\vspace{0.1cm}
\noindent \textbf{Contributions.} In this paper, we formally study multilevel fair allocation problems, assuming that internal nodes follow a utilitarian welfare objective and that leaves have additive preferences, and focusing on weighted envy-based fairness notions. We first show that WEF1 must be adapted to the multilevel setting, and that the choice of adaptation is not neutral. We propose three adaptations, ranging from the weakest to the strongest: pessimistic, agnostic, and optimistic. Under identical preferences, our three adaptations coincide and can be guaranteed by a multilevel extension of the Weighted Round Robin \cite{Chakraborty_et_al_2020} (MWRR). However, results become more nuanced under general additive preferences. Although the MWRR guarantees the pessimistic adaptation, it fails to satisfy the stronger ones. Nevertheless, we show in Section~\ref{subsection: numerical results} that the algorithm still performs extremely well experimentally with respect to the agnostic fairness notion.

\vspace{-0.2cm}
\section{Model}
\vspace{-0.2cm}

In this paper, we consider an allocation problem where a set of goods $\items = \{g_1, \ldots, g_m\}$ must be distributed among agents organized in a hierarchical structure.
\vspace{0.1cm}

\noindent \textbf{Hierarchical structure.} We consider a multilevel allocation problem represented by an {\it arborescence} (i.e. a directed rooted tree in which all edges point away from the root) denoted by $\tree = (\nodes, \edges)$, where $\nodes = \{1, \ldots, n\}$ is the set of nodes representing agents and $\edges$ is the set of arrows representing hierarchical relations among agents. We assume that the nodes are indexed according to a topological ordering of $\tree$ (hence the root of $\tree$ is node $1$). We also assume that any node $i \in \nodes$ has a weight $w_i \in \mathbb{R}_{> 0}$, which can be arbitrarily fixed, or can depend on the structure of the tree. For any node $i \in \nodes$, let $h(i)$ be the height of $i$, i.e. the number of arrows in a longest path starting at node $i$. Let $\children(i)$ denote the set of children of $i$, which is defined by $\children(i)\! = \!\{j \!\in \!\nodes \!: \!(i, j) \!\in \!\edges\}$. Let $\parent(i)$ be the parent of $i$, which is the unique node such that $(\parent(i),i) \in \edges$, and $\parent(1) = \emptyset$. Moreover, let $\ancestors(i)$ be the set of its ancestors, i.e. nodes belonging to the unique path from $1$ to $i$ (including $i$ itself). For any node $i$, let $\tree_i = (\nodes_i, \edges_i)$ denote the subtree of $\tree$ rooted at $i$, consisting of all nodes and edges belonging to paths that start at $i$. Let $\leaves(i)$ be the set of leaves of $\tree_i$, formally defined as $\leaves(i) = \{j \in \nodes_i : \children(j) = \emptyset\}$. Let $\internal(i)$ denote the set of internal nodes of $\tree_i$, defined by $\internal(i) = \nodes_i \setminus \leaves(i)$. In particular, $\leaves$ and $\internal$ denote the sets of leaves and internal nodes of the entire tree. For brevity, we write $\internal := \internal(1)$ and $\leaves := \leaves(1)$ throughout the paper, and use $\internal(i)$ and $\leaves(i)$ when referring to a specific subtree.

\vspace{0.1cm}
\noindent \textbf{Allocations.} We introduce different types of allocations relevant to our setting:

\begin{definition}[Complete multilevel allocation]
$\pi : \nodes \rightarrow 2^\items$ is a \emph{multilevel allocation} if it satisfies: (i) $\pi(1) = \items$, (ii) $\pi(i) \supseteq \cup_{j \in \children(i)} \pi(j), \forall i \in \internal$, and (iii) $\pi(i) \cap \pi(j) \! =\! \emptyset,  \forall i, j \in \nodes$ such that $\parent(i)\! = \! \parent(j)$, where $\pi(i)$ denotes the bundle of any node $i \in \nodes$. A multilevel allocation $\pi$ is \emph{complete} if condition (ii) is an equality. 
\end{definition}

\noindent We require that the root owns all items (i.e. $\pi(1) = \items$) only to ensure that none are discarded a priori. Moreover, we require that each internal node allocate each of its item to at most one child.
The set of all multilevel allocations is denoted by $\mallocations$ hereafter.

\begin{definition}[Restricted multilevel allocation] Given a multilevel allocation $\pi \in  \mallocations$ and a set of nodes $N \subseteq \nodes$, the \emph{restriction of $\pi $ to the nodes in $N$} is denoted by $\pi |_N$ and defined by $\pi |_N   =  (\pi(i))_{i \in N}$.
\end{definition}

\begin{definition}[Local allocation] Given a set of nodes $N \subseteq  \nodes$ and a set of items $S  \subseteq  \items$, $A  :  N  \rightarrow  2^S$ is a \emph{local allocation} if (i) $\cup_{i \in N} A(i)  \subseteq  S$ and (ii) $A(i) \cap A(j)  =  \emptyset$ for all $i, j  \in  N$. Such a local allocation $A$ is complete if condition (i) is an equality. 
\end{definition}

\noindent For any $N  \subseteq  \nodes$ and $S  \subseteq  \items$, let $\allocations^S_N$ denote the set of corresponding complete local allocations. Note that, for any complete multilevel allocation $\pi  \in  \mallocations$ and any node $i  \in  \internal$, the restricted allocation $\pi |_{\children(i)}$ is a complete local allocation belonging to $\allocations^{\pi(i)}_{\children(i)}$.

\vspace{0.1cm}

\noindent \textbf{Utility model.} Let $v_i : \mallocations \rightarrow \mathbb{R}_{\geq 0}$ be the utility function of node $i  \in  \nodes$, and let $v =(v_i)_{i \in \nodes}$. For any internal node $i  \in  \internal$ and any multilevel allocation $\pi  \in \mallocations$, $v_i(\pi)$ quantifies some concept of overall welfare derived by the children of $i$ from $\pi $.  We focus here on the utilitarian social welfare, i.e.
\setlength{\abovedisplayskip}{3pt}
\setlength{\belowdisplayskip}{2pt}
\[
v_i(\pi) = \sum_{j \in \children(i)} v_{j}(\pi)
\]

Hence, we assume that internal nodes have additive utilities over their children. 
Note that, by linearity of summation, $v_i(\pi)$ can be rewritten as $v_i(\pi) = \sum_{x \in \leaves(i)} v_x(\pi)$, where the sum ranges over the leaves of the subtree $\tree_i$. In contrast, leaves have no children and are here equipped with standard additive utilities: any leaf $x \in \leaves$ is associated with a utility function $u_x: 2^\items \rightarrow \mathbb{R}_{\geq 0}$ such that $v_x(\pi) = u_x(\pi(x))$ for any multilevel allocation $\pi \in \mallocations$. Formally, $u_x(\emptyset)=0$ and $u_x(S)=\sum_{g \in S} u_x(g)$ for any bundle $S \subseteq \items$, where $u_x(g)$ is a slight abuse of notation for $u_x(\{g\})$.%, the utility of leaf $x$ for item $g$. 

\vspace{0.1cm}

\noindent \textbf{Estimated utility functions. } Our goal is to compute a multilevel allocation that is both complete and fair. We focus on envy-based fairness notions, which typically require that each agent values its own bundle at least as much as any other agent’s bundle. Such notions rely on agents being able to evaluate bundles directly. However, in our setting, utility functions $v_i, i\in \internal$, are defined over multilevel allocations rather than bundles of items. As a result, standard envy-based notions cannot be applied directly. 

We therefore begin by introducing a method to evaluate bundles. 
As highlighted by \cite{Scarlett_et_al2023}, there are two ways to do so: (i) we can assume that $i$ computes an allocation of its bundle $S$ within $\tree_i$, in which case its estimated utility is the sum of its children’s (equivalently, leaves’) utilities, or (ii) we can assume that $i$ is agnostic to the allocation, and its estimated utility is the weighted average of $S$ over the leaves in $\leaves(i)$.
In this paper, we investigate both propositions. For (i), we propose two natural approaches: the first one, which we call the \textit{optimistic} approach, assumes that $i$ computes a utilitarian-optimal allocation of $S$ to its leaves $\leaves(i)$, while the \textit{pessimistic} approach assumes that $i$ compute an allocation of $S$ to $\leaves(i)$ minimizing the utilitarian-welfare. 

\vspace{-0.15cm}

\begin{definition}[Optimistic estimated utility function]
    Given a node $i \in \nodes$ and a subset of items $S \subseteq \items$, the \emph{optimistic estimated utility function} of $i$ for $S$ is 

\vspace{-0.25cm}

    $$\overset{\wedge}{v}_i(S) = \max_{A \in \allocations^S_{\leaves(i)}} \sum_{x \in \leaves(i)} \sum_{g \in A(x)} u_x(g)$$
\end{definition}

\vspace{-0.25cm}

\begin{definition}[Pessimistic estimated utility function]
    Given a node $i \in \nodes$ and a subset of items $S \subseteq \items$, the \emph{pessimistic estimated utility function} of $i$ for $S$ is 

    \vspace{-0.25cm}
    
    $$\overset{\vee}{v}_i(S) = \min_{A \in \allocations^S_{\leaves(i)}} \sum_{x \in \leaves(i)} \sum_{g \in A(x)} u_x(g)$$
\end{definition}

\vspace{-0.15cm}

For (ii), we refer to the weighted average evaluation of bundle $S$ as the \textit{agnostic} approach. On the contrary of the optimistic and pessimistic approaches, the agnostic approach does not assume a specific allocation to the leaves, and measures the value of a bundle as a weighted average utility over all leaves in $\leaves(i)$.

\vspace{-0.15cm}

\begin{definition}[Agnostic estimated utility function]
    Given a node $i \in \nodes$ and a subset of items $S \subseteq \items$, the \emph{agnostic estimated utility function} of $i$ for $S$ is 
    
    \vspace{-0.25cm}
    
    $$\overline{v}_i(S) = \sum_{x \in \leaves(i)} \sum_{g \in S} u_x(g) \times W(x, i) \mbox{ with } W(x, i) = \Pi_{k \in \ancestors(x) \setminus \ancestors(i)} \frac{w_k}{w_{\children(\parent(k))}}$$ where $w_{\children(\parent(k))} = \sum_{k' \in \children(\parent(k))} w_{k'}$.
\end{definition}

\vspace{-0.1cm}

For any internal node $i \in \internal$ and leaf $x \in \leaves(i)$, $W(x,i)$ can be interpreted as the weight of $x$ in the subtree $\tree_i$. The quantity $W(i,j)$ is defined similarly for any node $i \in \nodes$ and ancestor $j \in \ancestors(i)$. Moreover, Appendix~B shows that $\sum_{x \in \leaves(i)} W(x,i)=1$.

Note that, for any leaf $x \in \leaves$, we have $\overset{\wedge}{v}_x(S) = \overline{v}_x(S) = \overset{\vee}{v}_x(S) = u_x(S)$ for any bundle $S \subseteq \items$, since leaf $x$ is a tree with only one node (and thus $\leaves(x) = \{x\}$). 

\vspace{0.2cm}

\noindent \textbf{Fairness.} Having defined the estimation functions, we can propose some multilevel extensions of some envy-based notions. This definition is parametrized to accommodate our three estimated utility functions.

\vspace{-0.2cm}

\begin{definition}[M{[$\mathcal{E}$]}-WEF] \label{def: m-wef1}
    Multilevel allocation $\pi \! \in \! \mallocations$ is \emph{estimated Multilevel Weighted Envy-Free for}  $\mathcal{E} \in \{\text{pess}, \text{agno}, \text{opt}\}$ (denoted by M[$\mathcal{E}$]-WEF), if for any internal node $i \in \internal$, and any pair of children $j, k \in \children(i)$, we have 
    $$\frac{v_j(\pi)}{w_j} \geq \frac{\tilde{v}_j(\pi(k))}{w_k}$$

    \vspace{-0.1cm}
    
   \noindent where $\tilde{v}$ denotes $\overset{\wedge}{v}$, $\overline{v}$ or $\overset{\vee}{v}$ for $\mathcal{E}=\text{pess}$, $\text{agno}$, or $\text{opt}$, respectively.
    %some estimated utility function $\tilde{v}$.
\end{definition}

The definition of M[$\mathcal{E}$]-WEF1 is  obtained by simply requiring $\frac{v_j(\pi)}{w_j} \geq \frac{\tilde{v}_j(\pi(k) \setminus \{g\})}{w_k}$ for some item $g \in \pi(k)$. For  M[$\mathcal{E}$]-WEFX, this must hold  for all items $g \in \pi(k)$.

\medskip

Note that, if $\tree$ has height $h(1)\!\!=\!\!1$, the problem reduces to a standard \emph{monolevel} allocation setting, and our three estimated fairness notions coincide with the classical WEF, WEF1, and WEFX notions.

\begin{example} \label{example: example1}
    We now illustrate the different estimated utility functions. Assume an instance with $\nodes = \{1, \ldots, 7\}$ organized as in Fig~\ref{fig:initial_tree}, and $\items = \{g_1, \ldots, g_5\}$. Leaves 4 and 6 have the following preferences $u_x(g) = 2$ for $x \in \{4, 6\}$ and any $g \in \items$ ; leaves 5 and 7 have the following preferences $u_x(g) = 1$ for $x \in \{5, 7\}$ and any $g \in \items$. Assume the weight of any node $i \in \nodes$ is $w_i = |\leaves(i)|$. Suppose we have the multilevel allocation $\pi \in \mallocations$ such that $\pi(4) = \{g_1, g_5\}$, $\pi(5) = \{g_3\}$, $\pi(6) = \{g_2\}$, and $\pi(7) = \{g_4\}$. By definition, $\pi(2) = \{g_1, g_3, g_5\}$ and $\pi(3) = \{g_2, g_4\}$.

    To assess whether node $3$ is envious towards node $2$, we need to choose an estimated utility function to estimate the value of $2$'s bundle. Depending on which estimated utility functions we choose, the results might differ. Indeed, in this example, we can see that: $\overset{\wedge}{v}_3(\pi(2) \setminus \{g_1\}) = 4$ since all items are allocated to leaf $6$ in a utilitarian-optimal allocation ; $\overline{v}_3(\pi(2) \setminus \{g_1\}) = 3$ ; and finally, $\overset{\vee}{v}_3(\pi(2) \setminus \{g_1\}) = 2$ since all items are allocated to leaf $7$ to minimize the utilitarian-welfare. Hence, the conclusions on whether $3$ envies $2$ would differ based on which notion you use: according to the optimistic function, $3$ is weighted envious towards $2$ even up to one good, while according to both agnostic and pessimistic, $3$ is weighted envy-free towards $2$ up to one good. 
\end{example}

\section{Identical additive valuations} \label{section: identical valuations}

In \cite{Scarlett_et_al2023}, the authors propose algorithms achieving strong fairness guarantees under two assumptions: (i) all agents share identical preferences, and (ii) agents within the same group share identical preferences. We show that their algorithms and arguments extend naturally to our multilevel setting. In our framework, these assumptions translate to: (i) all leaves $x \in \leaves$ have identical preferences, and (ii) all leaves within each subtree rooted at a child of the root (i.e., $x \in \leaves(i)$ for $i \in \children(1)$) have identical preferences; 
for instance, in our food charity example, branches within the same department may face similar needs and thus share preferences over supplies. 
We refer to (i) as \emph{all-common} valuations, and to (ii) as \emph{root-child-common} valuations.

\begin{remark}
    Under both all-common and root-child-common valuations, for any internal node $i \in \internal \setminus \{1\}$ and any bundle $S \subseteq \items$, we have:
    \[ \forall A, B \in \allocations^S_{\leaves(i)}, 
        \sum_{x \in \leaves(i)} u_x(A(x)) = \sum_{x \in \leaves(i)} u_x(B(x)).
    \]
In other words, the allocation of $S$ to the leaves in $\leaves(i)$ does not affect the utility derived by  $i$. Consequently, the three estimated utility functions coincide, i.e.,  $\overset{\wedge}{v}_i(S) \! =\! \overline{v}_i(S)\! = \! \overset{\vee}{v}_i(S)$. Hence the definitions of M[$\mathcal{E}$]-WEF, for $\mathcal{E} \in$ \{pess, agno, opt\}, also coincide.

\end{remark}

Accordingly, in this section, we slightly abuse notation and refer to the value of a bundle even for internal nodes, rather than the value of a multilevel allocation. We also write M-WEF instead of M[$\mathcal{E}$]-WEF for brevity.

\subsection{All-common valuations}

\begin{theorem}
    Under all-common valuation, a M-WEFX allocation can be computed in polynomial time.
\end{theorem}

\begin{proof}
    The following algorithm is a multilevel extension of the algorithm presented in the proof of Theorem~3.1 in \cite{Scarlett_et_al2023}. Order the goods in decreasing order of preferences. Starting at the root, select the child with the least weighted bundle value (i.e. the utility of the bundle divided by the weight of the node). If this child is not a leaf, repeat the selection process until selecting a leaf, denoted $x \in \leaves$. Allocate the first remaining item (i.e. the highest-value remaining item) to the selected leaf. At any point in time, the leaf we selected (or any of its ancestors) cannot be weighted envied before the allocation of the new item by one of its siblings (as it has the least valued bundle). Any envy that forms towards any of the nodes $i \in \ancestors(x)$ can only result of the latest good allocated, and any envy will disappear upon dropping this good. Moreover, this good is also the least valued one in $i$'s bundle, by construction. The resulting multilevel allocation is thus M-WEFX. Furthermore, the algorithm is polynomial: sorting the goods is polynomial-time doable, the leaf-selection process is in $O(n)$, and allocating the chosen item at an iteration also takes $O(n)$. 
\end{proof}

We then show that we can obtain M-WEF1 jointly with some (monolevel) fairness notion at the leaves. The full proof and the algorithm can be found in the appendix. One component of the algorithm presented in the appendix is the Multilevel Weighted Round Robin, presented in Algorithm~\ref{algo: mwrr} in Section~\ref{subsection: root_child_common}. 

\begin{theorem}
    Under all-common valuations and $w_i = |\leaves(i)|$ for all $i \in \nodes$, there exists a polynomial-time algorithm that returns an allocation that is M-WEF1 and EFX between all leaves in $\leaves$.
\end{theorem}

\subsection{Root-child-common valuations} \label{subsection: root_child_common}

We now present a multilevel extension of the Weighted Round Robin (WRR) algorithm from \cite{Chakraborty_et_al_2020,Scarlett_et_al2023}. The principle of our algorithm, which we call Multilevel Weighted Round Robin (MWRR), is the following: each node $i \in \nodes$ is equipped with a picking score $t_i$ which counts the number of times $i$ was picked by the MWRR, and this picking score is weighted by $i$'s weight. Then, starting at the root, MWRR picks the child $i \in \children(1)$ minimizing this picking score (breaking ties lexicographically). If $i$ is an internal node, we repeat the selection process, i.e. we select the child $j \in \children(i)$ which minimizes the picking score, until the selected node is a leaf. Once it is the case, the selected leaf gets to choose her preferred item among the remaining ones. We repeat this procedure until all items are allocated. For the leaf selection procedure, ties are broken differently depending on whether they occur between internal nodes or leaves: (i) among internal nodes, ties are broken lexicographically; (ii) among leaves, they are broken in favor of the leaf with the highest utility for any remaining item. If the children of the considered node is a mix of internal nodes and leaves, we use (i). Pseudocode can be found in Algorithm~\ref{algo: mwrr}.
% differently depending on whether they occur between internal nodes or leaves: (i) among internal nodes, ties are broken lexicographically; (ii) among leaves, they are broken in favor of the leaf with the highest utility for any remaining item. If the children of the considered node is a mix of internal nodes and leaves, we use (i). Pseudocode can be found in Algorithm~\ref{algo: mwrr}.

\begin{algorithm}[!h] \small
\caption{\small Multilevel Weighted Round Robin (MWRR)}
\label{algo: mwrr}
\begin{algorithmic}[1]
    \State \textbf{Input:} $\tree$ - a multilevel tree ; $\items$ - a set of items ; the leaves' valuations $(u_x)_{x \in \leaves}$
    \State \textbf{Output:} $\pi $ - a multilevel allocation
    \State Initialize $\pi$ with empty bundle for any $i \in \nodes \setminus \{1\}$ and $\pi(1) = \items$
    \State Initialize picking scores s.t. $t_i = 0, \forall i \in \nodes$
    \State $\mathcal{RI} \gets \items$ \Comment{Initial set of remaining items}
    \While{$\mathcal{RI} \neq \emptyset$}
        \State $i = \arg \min_{i' \in \children(1)} \frac{t_{i'}}{w_{i'}}$
        \While{$i \notin \leaves$}
            \State $i = \arg \min_{i' \in \children(i)} \frac{t_{i'}}{w_{i'}}$
        \EndWhile
        \State $g = \arg \max_{g' \in \mathcal{RI}} u_i(g')$
        \While{$i \neq 1$}
            \State $\pi(i) \gets \pi(i) \cup \{g\}$
            \State $t_i \gets t_i + 1$
            \State $i \gets \parent(i)$
        \EndWhile
         \State $\mathcal{RI} \gets \mathcal{RI} \setminus \{g\}$
    \EndWhile
    \State \Return $\pi$
\end{algorithmic}
\end{algorithm}

Our goal is to study what fairness properties MWRR may guarantee. We first look at its performance under root-child-common additive valuations, and in Section~\ref{section: general additive valuations}, we study thoroughly how fair it is in the general additive case. 

\begin{theorem} \label{theorem: mwrr mwef1}
    Under root-child-common valuations and arbitrary weights, the MWRR returns an M-WEF1 allocation in polynomial time.
\end{theorem}

\begin{proof}
    We first show that the MWRR runs in polynomial time: the while loop (Line 6) runs in $O(m)$ steps. Finding the child of the root with minimum weighted picking score (Line 7) is at most in $O(n)$ (for tree of height 1), and a loose bound for the while loop (Line 8) is $O(n^2)$. Finding the item with maximum utility (Line 11) can be done in $O(m)$, and updating the bundle (Line 12) requires at most $O(n)$ (in a comb tree). Hence, the complexity of the algorithm is $O(m(n^2 + m))$.

    Then, recall that for any child of the root, $i \in \mathcal{C}(1)$, all leaves in $\leaves(i)$ have the same preferences over singletons, i.e. $\forall x, y \in \leaves(i), \forall g \in \items, u_x(g) = u_y(g)$. Hence, any child of the root $i$ can be seen as an agent with additive utility over items: no matter which of its leaves is chosen at any iteration, node $i$ will receive the same item (as all leaves in $\leaves(i)$ have the same preferences), which will yield the same utility. Moreover, MWRR at the root chooses a child according to the "least weight-adjusted frequent picker" criterion, exactly like in \cite{Chakraborty_et_al_2020}. Since their algorithm is known to be WEF1 w.r.t. agents with additive utilities over items, we can conclude that MWRR is WEF1 w.r.t. $\mathcal{C}(1)$. Then, we can repeat the same argument recursively for any node $i \in \nodes$. 
\end{proof}

\begin{corollary}
    Under root-child-common valuations and $w_i = |\leaves(i)|$  for all $i \in \nodes$, the MWRR is M-WEF1 and EF1 between all leaves in $\leaves$.
\end{corollary}

\begin{proof}
    From Theorem~\ref{theorem: mwrr mwef1}, we know that MWRR satisfies M-WEF1. Furthermore, the proof of Theorem~3.9 in \cite{Scarlett_et_al2023} can be extended to our setting and establishes EF1 among the leaves. The full proof can be found in Appendix~A. 
\end{proof}

In this section, we showed that the algorithms proposed in \cite{Scarlett_et_al2023} could be extended easily to satisfy our multilevel fairness properties, namely M-WEF1 and M-WEFX, sometimes even jointly with fairness at the leaves. In Section~\ref{section: general additive valuations} we discuss how to extend our multilevel fairness notions under general additive valuations. 

\section{General additive valuations} \label{section: general additive valuations}

We now focus on the more general case where we have general additive valuations.

\subsection{Relations among M[$\mathcal{E}$]-WEF1 notions}

We formally establish the hierarchy between the adaptations we presented. 

\begin{lemma} \label{lemma: opt > agno}
    For any internal node $i \in \internal$ and bundle $S \subseteq \items$, we have $\overset{\wedge}{v}_i(S) \geq \overline{v}_i(S)$.
\end{lemma}

\begin{proof}
By linearity, $\overline{v}_i(S)$  and $\overset{\wedge}{v}_i(S)$ can be rewritten as follows: 
\begin{align*}\small
\overline{v}_i(S) &= \sum_{x \in \leaves(i)} \sum_{g \in S} u_x(g)\, W(x,i) = \sum_{g \in S} \sum_{x \in \leaves(i)} u_x(g)\, W(x,i)\\
\overset{\wedge}{v}_i(S) &= \max_{A \in \allocations^S_{\leaves(i)}} \sum_{x \in \leaves(i)} \sum_{g \in A(x)} u_x(g)
= \sum_{g \in S} \max_{x \in \leaves(i)} u_x(g)
\end{align*}
Note that, for any $g \in S$, both  $\sum_{x \in \leaves(i)} u_x(g)\, W(x,i)$ and $\max_{x \in \leaves(i)} u_x(g)$ are convex combinations of $(u_x(g))_{x \in \leaves(i)}$; the latter places all the weight on some leaf $x^* \in \arg\max_{x \in \leaves} u_x(g)$, while the former uses weights $(W(x,i))_{x \in \leaves(i)}$, which satisfy \linebreak $\sum_{x \in \leaves(i)} W(x,i) = 1$ (as shown in Appendix~B). Since a convex combination is maximized by putting all the weight on the largest element, we obtain:
\[
\max_{x \in \leaves(i)} u_x(g)
\geq
\sum_{x \in \leaves(i)} u_x(g)\, W(x,i).
\]
Then, summing over $g \in S$ gives:
\[
\overset{\wedge}{v}_i(S)
= \sum_{g \in S} \max_{x \in \leaves(i)} u_x(g)
\geq
\sum_{g \in S} \sum_{x \in \leaves(i)} u_x(g)\, W(x,i)
= \overline{v}_i(S).
\] 
\end{proof}

\begin{lemma} \label{lemma: agno > pess}
    For any internal node $i \in \internal$ and bundle $S \subseteq \items$, we have $\overline{v}_i(S) \geq \overset{\vee}{v}_i(S)$.
\end{lemma}

% \vspace{-0.6cm}

\begin{proof}
    The proof is similar to that of Lemma~\ref{lemma: opt > agno}: since $\overset{\vee}{v}_i(S)$ is obtained by placing all weight on the leaf with minimum utility for $g$, we have $\sum_{x \in \leaves(i)} u_x(g) W(x,i) \ge \min_{x \in \leaves(i)} u_x(g)$, and thus $\overline{v}_i(S) \ge \overset{\vee}{v}_i(S)$. 
\end{proof}

These relationships between the estimated utility functions induce an implication structure among the fairness notions. The proof trivially results from Lemmas~\ref{lemma: opt > agno} and \ref{lemma: agno > pess}.

\begin{proposition}
$\text{M[opt]-WEF1} \;\Rightarrow\; \text{M[agno]-WEF1} \;\Rightarrow\; \text{M[pess]-WEF1}$.
\end{proposition}

In other words, for any allocation $\pi \in \mallocations$, if $\pi$ is not M[pess]-WEF1, then it is not M[agno]-WEF1; similarly, if it is not M[agno]-WEF1, then it is not M[opt]-WEF1. However, some allocations that are not M[opt]-WEF1 may still be M[agno]-WEF1 (e.g., the allocation in Example~\ref{example: example1}), and some that are not M[agno]-WEF1 may still be M[pess]-WEF1, as shown in the following example.

\begin{example}
    We can slightly modify the instance of Example~\ref{example: example1} to exhibit a multilevel allocation that is M[pess]-WEF1 but not M[agno]-WEF1. We now have $\nodes = \{1, \ldots, 9\}$ organized in the tree of Fig.~\ref{fig:example3-tree}. Leaves $x = 4, 5$ and $7$ have utilities $u_x(g) = 2$ for all $g \in \items$ ; leaves $x=6,8$ and $9$ have utilities $u_x(g) = 1$ for all $g \in \items$. Assume the weight of any node $i \in \internal$ is $w_i = |\leaves(i)|$, and that of any leaf $x \in \leaves$ is 1. Suppose we have the multilevel allocation $\pi \in \mallocations$ such that $\pi(4) = \{g_1\}$, $\pi(5) = \{g_3\}$, $\pi(6) = \{g_5\}$, $\pi(7) = \{g_2\}$, $\pi(8) = \{g_4\}$, and $\pi(9) = \emptyset$.
    For such an allocation $\pi$, we  have $v_3(\pi(3))= 3$, and for any $g\in \pi(2)$, $\overline{v}_3(\pi(2) \setminus \{g\}) = 10/3$, while $\overset{\vee}{v}_3(\pi(2) \setminus \{g\}) = 2$. Hence, since $w_2=w_3$, $\pi$ is M[pess]-WEF1, but not M[agno]-WEF1.
\end{example}

\vspace{-1.4cm}

\begin{figure}[h!]
\begin{center}
\begin{minipage}{0.48\linewidth}
\resizebox{\linewidth}{!}{
\begin{tikzpicture}[
    level distance=7mm,
    every node/.style={circle, draw, inner sep=2pt, font=\footnotesize},
    level 1/.style={sibling distance=30mm},
    level 2/.style={sibling distance=10mm},
    level 3/.style={sibling distance=6mm},
    edge from parent/.style={draw, -stealth}
]

\node (1) {1}
    child { node (2) {2}
        child { node (4) {4} }
        child { node (5) {5} }
        child { node (6) {6} }
    }
    child { node (3) {3}
        child { node (7) {7} }
        child { node (8) {8} }
        child { node (9) {9} }
    };

\end{tikzpicture}
}
\caption{Tree of Example 2.}
\label{fig:example3-tree}
\end{minipage}
\end{center}
\end{figure}
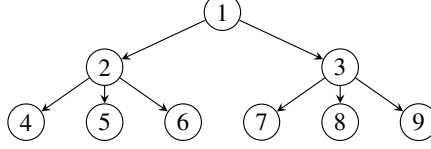

\vspace{-1.7cm}

\subsection{On the existence of M[$\mathcal{E}$]-WEF1 allocation}

In this section, we show that MWRR guarantees M[pess]-WEF1 under general additive valuations. In contrast, M[agno]-WEF1 is harder to ensure: some instances admit no such allocation (and thus no M[opt]-WEF1 allocation either), and MWRR may fail to guarantee it even when one exists. Nevertheless, Section~\ref{subsection: numerical results} shows that such failures remain rare in practice.

We begin by showing that MWRR guarantees M[pess]-WEF1. Our approach relies on a key lemma, which allows us to adapt the proof of WEF1 for the Weighted Round Robin algorithm from \cite{Chakraborty_et_al_2020} to our multilevel setting.

\begin{theorem}
    MWRR always returns an M[pess]-WEF1 allocation.
\end{theorem}

\vspace{-0.2cm}
\begin{proofsketch}
The proof follows the same structure as the WRR proof of \cite{Chakraborty_et_al_2020}, with an additional lemma (Lemma~5, Appendix~B) to handle the multilevel pessimistic setting. We first prove the result for the children of the root. The same argument then applies recursively to any internal node. During the first $|\children(1)|$ steps, each child of the root receives at most one item. Hence, the allocation is trivially M[pess]-WEF1 at this stage. Then, as in WRR, the picking rule ensures that for any $i,j \in \children(1)$, we have $\frac{t_j}{t_i} \ge \frac{w_j}{w_i}$.
We show a lemma stating that whenever node $j$ is selected, the utility it derives from the chosen item (through its selected leaf) is at least the worst utility that any leaf in $\leaves(i)$ can obtain from any item allocated in subsequent steps. In particular, it dominates the worst item eventually received by any sibling $i$. Combining the picking ratio property with the mentioned lemma, we obtain that the total utility accumulated by $j$ (excluding its first item) is at least the total worst-case utility of $i$. This implies that $j$ is M[pess]-WEF1 with respect to $i$.
Applying the same argument recursively to each internal node concludes the proof.

\end{proofsketch}

\begin{theorem} \label{theorem: inexistence agnostic}
    Under general valuations, an M[agno]-WEF1 allocation needs not exist.
\end{theorem}

\begin{proof}
Consider the multilevel instance represented in Fig.~\ref{fig:example4-tree}, and assume weights are $w_i = |\leaves(i)|$ for $i \in \nodes$. Let $\items=\{g_1,\ldots,g_5\}$. Leaves have additive utilities: children of $4$ and $6$ value every item at $2$, while children of $5$ and $7$ value every item at $1$. We show that no multilevel allocation can be M[agno]-WEF1.

\vspace{-0.7cm}
\begin{figure}[H]
\begin{center}
\begin{minipage}{0.5\linewidth}
\centering
\resizebox{\linewidth}{!}{
\begin{tikzpicture}[
    level distance=13mm,
    every node/.style={circle, draw, inner sep=2pt, font=\Large},
    level 1/.style={sibling distance=55mm},
    level 2/.style={sibling distance=28mm},
    level 3/.style={sibling distance=9mm},
    edge from parent/.style={draw, -stealth}
]
\node {1}
    child { node {2}
        child { node {4}
            child { node {8} }
            child { node {9} }
            child { node {10} }
            child { node {11} }
        }
        child { node {5}
            child { node {12} }
            child { node {13} }
        }
    }
    child { node {3}
        child { node {6}
            child { node {14} }
            child { node {15} }
            child { node {16} }
            child { node {17} }
        }
        child { node {7}
            child { node {18} }
            child { node {19} }
        }
    };
\end{tikzpicture}
}
\caption{Tree of Example 3}
\label{fig:example4-tree}
\end{minipage}
\end{center}
\end{figure}
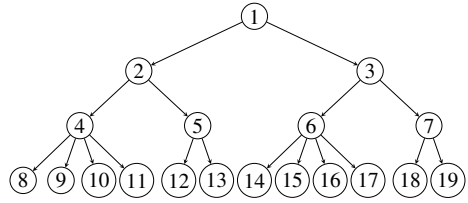
\vspace{-1.2cm}

First, notice that no multilevel allocation allocating all items to (leaves of) node $2$ can be M[agno]-WEF1. Indeed, in such case, $v_3(\pi)/w_3 = 0 < \overline{v}_3(\items \setminus \{g\})/w_2 = 20/18$ for any item $g \in \items$. Conversely, no multilevel allocation allocating all items to leaves of node $3$ can be M[agno]-WEF1, by symmetry of the instance.

Similarly, no multilevel allocation allocating exactly one item $g \in \items$ to node $3$ can be M[agno]-WEF1. Indeed, even if this item was allocated to a leaf in $\children(6)$ (i.e. the group of leaves with the highest utilities), we would have $v_3(\pi)/w_3\! =\! 2/6\! < \!\overline{v}_3(\pi(2) \backslash \{g'\})\! =\! 10/12$, where $\pi(2)\! =\! \items \setminus \{g\}$, and $g' \! \in \! \pi(2)$. The same applies when allocating exactly one item to node $2$, and the rest to node $3$, again by symmetry.

Finally, we show that even allocating $3$ items to node $2$, and $2$ items to node $3$ (or conversely) cannot yield a M[agno]-WEF1 allocation. Suppose a multilevel allocation where $|\pi(2)| = 3$ and $|\pi(3)| = 2$. Then, we show that no matter how $\pi(3)$ is allocated to its leaves, it must be agnostic weighted envious towards node $2$. Indeed, assume the two items are allocated to leaves in $\children(6)$, then node $7$ is agnostic weighted envious towards node $6$, since it received no item, and node $6$ received $2$ items. Then, node $3$ is obliged to allocate one item to node $6$, and one item to node $7$. But then, we have $v_3(\pi)/w_3 = 3/6 < \overline{v}_3(\pi(2) \setminus \{g\})/w_2 = 10/18$. By symmetry, no multilevel allocation allocating $2$ items to node $2$, and $3$ items to node $3$ can be M[agno]-WEF1. 

\end{proof}

This example shows that under general valuations, an M[agno]-WEF1 allocation may fail to exist. However, this is not the reason why MWRR fails: the algorithm may return an allocation that is not M[agno]-WEF1 even when such an allocation exists.

\begin{theorem} \label{theorem: mwrr not agno}
    Under general valuations, MWRR does not guarantee an M[agno]-WEF1 allocation, even when it exists.
\end{theorem}

\begin{example}

We consider the same instance as in the proof of Theorem~\ref{theorem: inexistence agnostic}, modifying only the leaves’ preferences over singletons. We assume that all leaves in $\children(5)$ and $\children(7)$ have utility $u_x(g)=1/5$ for any $g \in \items$. For $x \in \children(4)$, we set $u_x(g)=5/21$ for $g \in \{g_1,g_2,g_4,g_5\}$ and $u_x(g)=1/21$ for $g_3$. Finally, for $x \in \children(6)$, we set $u_x(g)=5/21$ for $g \in \{g_1,g_2,g_3,g_5\}$ and $u_x(g)=1/21$ for $g_4$.

For this instance, MWRR outputs multilevel allocation $\pi$, where we have : $\pi(8) = \{g_1\}, \pi(9) = \{g_5\}, \pi(12) = \{g_3\}, \pi(14) = \{g_2\}, \pi(18) = \{g_4\}$, and all other leaves received nothing. One can check that $v_3(\pi) = 46/105 \simeq 0.438$ and $\bar{v}_3(\pi(2) \setminus \{g_1\}) = 142/315 \simeq 0.451$. Since nodes 2 and 3 have the same weights, we can deduce that $\pi$ is not M[agno]-WEF1. However, there does exist an M[agno]-WEF1 allocation: $\pi(8) = \{g_2\}$, $\pi(12) = \{g_3\}$, $\pi(13) = \{g_4\}, \pi(14) = \{g_1\}, \pi(18) = \{g_5\}$.  
    
\end{example}

\begin{corollary}
    Under general valuations, MWRR is not guaranteed to compute an M[opt]-WEF1 allocation, even when it exists.
\end{corollary}

Though these results are negative, one may still ask whether an $\alpha$-approximation of M[agno]-WEF1 can be guaranteed. 

\begin{definition}[$\alpha$-M{[agno]}-WEF1]
    A multilevel allocation $\pi$ is an $\alpha$-approximation of M[agno]-WEF1, for some $\alpha > 0$, if, for any internal node $i \in \internal$ and any children $j,k \in \children(i)$, there exists $g \in \pi(k)$ such that:$$v_j(\pi)/w_j \ge \alpha \cdot \overline{v}_j(\pi(k)\setminus\{g\})/w_k$$ 
\end{definition}

\noindent Unfortunately, we show that MWRR cannot guarantee any such $\alpha$-approximation.

\begin{theorem} \label{theorem: inapproximability}
    For any constant factor $\alpha > 0$, there exists an instance where the allocation returned by MWRR is not an $\alpha$-approximation of M[agno]-WEF1.
\end{theorem}

\begin{proof}
Consider the tree in Fig.~\ref{fig:theorem_inapproximability_tree}, and assume $w_i = |\leaves(i)|$ for all  $i \in \nodes$. Let $\items = \{g_1, g_2, g_3\}$ be the set of items. Consider the following utilities: for any $x \in \leaves(2)$, $u_x(\{g_1\}) = 0$ and $u_x(\{g_2\}) = u_x(\{g_3\}) = 1$ ; for any $x \in \children(6)$, $u_x(\{g_1\}) = 1$ and $u_x(\{g_2\}) = u_x(\{g_3\}) = 0$ ; and for any $x \in \children(7)$, $u_x(\{g_1\}) = 0$ and $u_x(\{g_2\}) = u_x(\{g_3\}) = M$ for $M>0$ arbitrarily large.

Let $\pi$ be the allocation returned by MWRR on this instance: $\pi(8) = \{g_2\}, \pi(10) = \{g_3\}, \pi(12) = \{g_1\}$. Note that we have $v_3(\pi) = 1$, and $\overline{v}_3(\pi(2) \setminus \{g\}) = W(14, 3)\times M + W(15, 3) \times M = \frac{M}{2}$ for any $g \in \pi(2)$.  Hence, for $\pi$ to be an $\alpha$-approximation of M[agno]-WEF1, we need in particular $v_3(\pi)/w_3 \ge \alpha \cdot \overline{v}_3(\pi(2)\setminus\{g_1\})/w_2$. This requires $\alpha \le 2/M$. Since $M$ can be arbitrarily large, $\alpha$ can be arbitrarily small.
    
\end{proof}

 \vspace{-0.9cm}
 
\begin{figure}[h!]
\centering
\begin{tikzpicture}[
    scale=0.6,
    transform shape,
    level distance=12mm,
    every node/.style={circle, draw, inner sep=2pt, font=\Large},
    level 1/.style={sibling distance=50mm},
    level 2/.style={sibling distance=28mm},
    level 3/.style={sibling distance=9mm},
    edge from parent/.style={draw, -stealth}
]
\node {1}
    child { node {2}
        child { node {4}
            child { node {8} }
            child { node {9} }
        }
        child { node {5}
            child { node {10} }
            child { node {11} }
        }
    }
    child { node {3}
        child { node {6}
            child { node {12} }
            child { node {13} }
        }
        child { node {7}
            child { node {14} }
            child { node {15} }
        }
    };
\end{tikzpicture}
\caption{Tree of Theorem~\ref{theorem: inapproximability}}
\label{fig:theorem_inapproximability_tree}
\end{figure}
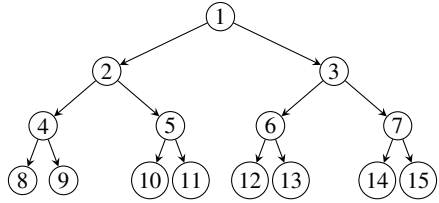
\vspace{-0.8cm}

This result is particularly striking when compared to the $\frac{1}{3}$-approximation of \cite{Scarlett_et_al2023} in the bilevel setting. It highlights a fundamental gap between bilevel instances and multilevel trees of height at least three. While the bilevel case still admits an $\alpha$-approximation for some $\alpha>0$, this is no longer possible in the multilevel setting.

\subsection{Experimental results} \label{subsection: numerical results}

In this section, we run experiments to test how often MWRR fails to return a fair allocation. We show that, while it may frequently violate the optimistic notion, it is almost always fair under the agnostic one (an experimental finding that significantly refines the negative result of Theorem~\ref{theorem: mwrr not agno}). We describe hereafter the experimental protocol (hardware details are provided in Appendix~C).

\vspace{0.1cm}

\noindent \textbf{Trees.} We consider three classes of trees in our experiments: (1) balanced binary trees, (ii) comb trees, and (iii) partially unbalanced trees, which are binary except at the last internal level: for any pair of siblings at this level, one internal node has two children while the other has five. For balanced and comb trees, we tested for $n=\{15, 31, 63, 127\}$, and for partially unbalanced trees, we tested for $n=\{21, 43, 87, 175\}$. Since the number of nodes varies between comb/balanced trees and partially unbalanced trees, we will refer to those number of nodes as \{small, medium-, medium+, large\}. In all cases, we tested for $m=\{n, 2n\}$ items. 
Moreover, we evaluate two weighting schemes: (i) each node $i \in \nodes$ is assigned weight $w_i = |\leaves(i)|$, or (ii) the weight of each node is sampled independently and uniformly at random from the integer interval $[1,6]$. In the tables below, we denote (i) by LW for leaf-count weights and (ii) by RW for random weights.

\vspace{0.1cm}

\noindent \textbf{Preference generation.}  To ensure a robust experimental evaluation, we generate preferences using four different methods. Specifically, we propose multilevel adaptations of (i) the Mallows model, (ii) the resampling Dirichlet model recently introduced by \cite{Bohm_et_al2026}, (iii) cost utilities \cite{Botan_et_al2023}, and (iv) correlated utilities as in \cite{Dickerson_et_al2014}. Importantly, our results are mostly consistent across all generation methods, which is why we do not devote extensive space to their detailed presentation in the main paper. The generation methods are nevertheless detailed in Appendix~C. 

\vspace{0.1cm}
\noindent \textbf{Results.} For each class of instances, we report the 95\% confidence interval for the proportion of cases in which the MWRR algorithm produces an M[agno]-WEF1 or an M[opt]-WEF1 allocation. We also provide the average running time of MWRR and its standard deviation. All results are based on 200 instances per class. Due to space constraints, some tables are deferred to the appendix.

\vspace{0.1cm}
\noindent \textit{Running time.} Experimental results show that MWRR is extremely fast and scales well (see Table~\ref{tab:runtime_all}). For instance, for balanced trees, the average running time for $n=15, m=15$ (over all generation methods, all weighting schemes) is 0.0002 second, while that for $n=127, m=254$ is 0.0191 second. Moreover, the running time is consistent across different types of instances.

 \vspace{-1cm}
\begin{table}[H]
    \footnotesize
    \centering
    \caption{Average running time (s) $\pm$ std.}
    \label{tab:runtime_all}
    \resizebox{0.8\textwidth}{!}{
    \begin{tabular}{c|c|c|c|c}
        \toprule
        $n$ & $m$ & \textbf{Comb} & \textbf{Balanced} & \textbf{Partially unbalanced} \\
        \midrule
        small & $\quad n \quad$ & $\quad0.0002 \pm 0.0001\quad$ & $\quad0.0002 \pm 0.0000\quad$ & $0.0003 \pm 0.0001$ \\
        small & $2n$ & $0.0003 \pm 0.0001$ & $0.0003 \pm 0.0001$ & $0.0007 \pm 0.0001$ \\
        medium- & $n$ & $0.0006 \pm 0.0003$ & $0.0005 \pm 0.0001$ & $0.0010 \pm 0.0001$ \\
        medium- & $2n$ & $0.0012 \pm 0.0006$ & $0.0011 \pm 0.0001$ & $0.0025 \pm 0.0003$ \\
        medium+ & $n$ & $0.0027 \pm 0.0018$ & $0.0017 \pm 0.0001$ & $0.0038 \pm 0.0003$ \\
        medium+ & $2n$ & $0.0057 \pm 0.0034$ & $0.0041 \pm 0.0003$ & $0.0099 \pm 0.0012$ \\
        large & $n$ & $0.0158 \pm 0.0126$ & $0.0067 \pm 0.0004$ & $0.0155 \pm 0.0015$ \\
        large & $2n$ & $0.0387 \pm 0.0296$ & $0.0191 \pm 0.0034$ & $0.0493 \pm 0.0085$ \\
        \bottomrule
    \end{tabular}
    }
\end{table}
 \vspace{-0.7cm}

\vspace{0.1cm}
\noindent \textit{M[agno]-WEF1.} Interestingly, our experiments strongly mitigate the impossibility results of Section~\ref{section: general additive valuations}. Although we cannot formally guarantee that MWRR always returns an M[agno]-WEF1 allocation, we observe that it does so in almost all cases. Out of the 192,000 generated instances, fewer than 50 result in allocations that are not M[agno]-WEF1. The few instances that led to unfair allocations were in general for comb trees with random weights (tables can be found in the appendix). This suggests that MWRR is extremely likely, in practice, to return a fair allocation in the agnostic sense.

\vspace{0.1cm}
\noindent \textit{M[opt]-WEF1.} In contrast, the optimistic notion of fairness appears significantly more challenging to satisfy (see Table~\ref{tab:fair2_all}). The performance of MWRR varies substantially depending on the tree structure, the size of the instance, and whether weights are randomly generated. For example, instances based on comb trees with large $n$ and random weights are particularly challenging, with sometimes near 100\% of computed allocations failing to satisfy M[opt]-WEF1. On the other hand, for balanced trees with weights defined as $w_i = |\leaves(i)|$ for all $i \in \nodes$, achieving M[opt]-WEF1 appears somewhat less challenging. Nevertheless, the proportion of unfair allocations remains highly variable and can still be substantial. Overall, the choice of weights (random or related to the number of leaves) emerges as the most influential factor affecting fairness performance.
 
\vspace{-1cm}
\begin{table}[H]
    \centering
    \footnotesize
    \caption{\centering Proportion of non-M[opt]-WEF1 allocations (95\% CI).}
    \label{tab:fair2_all}
    \resizebox{0.8\textwidth}{!}{
    \begin{tabular}{ll|cc|cc|cc}
        \toprule
        \multicolumn{2}{c}{} 
        & \multicolumn{2}{c}{\textbf{Comb}} 
        & \multicolumn{2}{c}{\textbf{Balanced}} 
        & \multicolumn{2}{c}{\textbf{Part. unbal.}} \\
        \cmidrule(lr){3-4} \cmidrule(lr){5-6} \cmidrule(lr){7-8}
        $n$ & $m$ 
        & RW & LW 
        & RW & LW 
        & RW & LW \\
        \midrule
        small & $n$ & $[0.25,\,0.27]$ & $[0.02,\,0.03]$ & $[0.02,\,0.03]$ & $[0.03,\,0.04]$ & $[0.05,\,0.06]$ & $[0.04,\,0.05]$ \\
        small & $2n$ & $[0.24,\,0.27]$ & $[0.01,\,0.02]$ & $[0.02,\,0.03]$ & $[0.02,\,0.03]$ & $[0.06,\,0.07]$ & $[0.04,\,0.05]$ \\
        medium- & $n$ & $[0.53,\,0.56]$ & $[0.02,\,0.03]$ & $[0.09,\,0.11]$ & $[0.09,\,0.11]$ & $[0.14,\,0.17]$ & $[0.12,\,0.14]$ \\
        medium- & $2n$ & $[0.57,\,0.60]$ & $[0.01,\,0.01]$ & $[0.08,\,0.09]$ & $[0.09,\,0.10]$ & $[0.14,\,0.16]$ & $[0.11,\,0.12]$ \\
        medium+ & $n$ & $[0.70,\,0.73]$ & $[0.01,\,0.02]$ & $[0.21,\,0.23]$ & $[0.21,\,0.23]$ & $[0.29,\,0.32]$ & $[0.22,\,0.25]$ \\
        medium+ & $2n$ & $[0.71,\,0.74]$ & $[0.01,\,0.01]$ & $[0.18,\,0.20]$ & $[0.12,\,0.15]$ & $[0.27,\,0.30]$ & $[0.16,\,0.18]$ \\
        large & $n$ & $[0.74,\,0.77]$ & $[0.01,\,0.01]$ & $[0.40,\,0.43]$ & $[0.37,\,0.40]$ & $[0.48,\,0.51]$ & $[0.38,\,0.41]$ \\
        large & $2n$ & $[0.91,\,0.94]$ & $[0.02,\,0.04]$ & $[0.55,\,0.61]$ & $[0.28,\,0.34]$ & $[0.62,\,0.69]$ & $[0.51,\,0.57]$ \\
        \bottomrule
    \end{tabular}
    }
\end{table}
\vspace{-1cm}

\section{Conclusion}
In this paper, we consider multilevel fair division problems where indivisible goods are allocated to agents organized in a tree-structured hierarchy. We assume that internal nodes have additive utilities over their children, while leaves have additive utilities over items. 
In this setting, we propose three envy-based fairness notions, depending on how a node estimates the value of a bundle, namely M[pess]-WEF1, M[agno]-WEF1, and M[opt]-WEF1. Under identical preferences, these notions coincide, and can be satisfied together with completeness using MWRR. Under general valuations, MWRR  guarantees only M[pess]-WEF1 but still performs very well in practice for M[agno]-WEF1. 

A direct extension concerns the existence of M[agno]-WEF1 allocations. In Section~\ref{section: general additive valuations}, we showed that under general additive valuations, such an allocation may not exist. However, this relies on non-normalized utilities: after normalization, the instance admits an M[agno]-WEF1 allocation, and existence under normalized utilities remains open. Still, this does not resolve the issue with MWRR, which may fail to return an M[agno]-WEF1 allocation even when one exists (as Example~\ref{fig:example3-tree} is already normalized).

A natural extension of this work is to consider broader utility functions for leaves, and alternative utilities for internal nodes. In particular, instead of purely aggregating children’s utilities, internal nodes could incorporate their own preferences, as in \cite{Bu_et_al2024}, e.g., preferring balanced allocations across subgroups.

\begin{credits}
\subsubsection{\ackname} This work is supported by the ANR project ANR-24-CE23-1700 AROMATICS.
\end{credits}

\bibliographystyle{splncs04}
\bibliography{bibliography}

\section*{A. Missing proofs from Section~3}

We start by proving Theorem~2.

\setcounter{theorem}{1}
\begin{theorem}
    Under all-common valuations and $w_i = |\leaves(i)|, \forall i \in \nodes$, there exists a polynomial-time algorithm that returns an allocation that is M-WEF1 and EFX between all leaves in $\leaves$.
\end{theorem}

\begin{proof}
    The algorithm is called Sequential Maximin - Multilevel Weighted Round Robin, and is a multilevel extension of the SM-IWRR presented in \cite{Scarlett_et_al2023}. Pseudocodes can be found in Algorithms~\ref{algo: sm} and \ref{algo: sm-mwrr}. The pseudocode for MWRR can be found in the main paper.

    \begin{algorithm}[!h] \small
    \caption{\small Sequential Maximin (from \cite{Scarlett_et_al2023})}
    \label{algo: sm}
    \begin{algorithmic}[1]
        \State \textbf{Input:} $N$ - a set of agents ; $\items$ - a set of items ; $u$ - a common valuation
        \State \textbf{Output:} $A$ - an allocation
        \State Initialize $A(x) = \{\}$ for $x \in N$
        \State $\mathcal{RI} \gets \items$ \Comment{Initial set of remaining items}
        \While{$\mathcal{RI} \neq \emptyset$}
            \State $g \gets \arg \max_{g' \in \mathcal{RI}} u({g'})$
            \State $x \gets \arg \min_{y \in N} u(A(y))$
            \State $A(x) \gets A(x) \cup \{g\}$
            \State $\mathcal{RI} \gets \mathcal{RI} \setminus \{g\}$
        \EndWhile
        \State \Return $A$
    \end{algorithmic}
    \end{algorithm}
    
    \begin{algorithm}[!h] \small
    \caption{\small Sequential Maximin - Multilevel Weighted Round Robin (SM-MWRR)}
    \label{algo: sm-mwrr}
    \begin{algorithmic}[1]
        \State \textbf{Input:} $\tree$ - a multilevel tree ; $\items$ - a set of items ; $u$ - a common valuation
        \State \textbf{Output:} $\pi$ - a multilevel allocation
        \State Run the SM algorithm with input $\leaves$, $\items$, and valuation $u$, and obtain allocation $A' \in \allocations^\items_{\leaves}$
        \State $x_{min} \gets \arg \min_{y \in \leaves(i)} u(A(y))$
        \State $A'_{min} \gets A'(x_{min})$
        \State Initialize a set of representative goods, $\mathcal{RG} = \{\}$
        \For{$x \in \leaves$}
            \State Create a representative good $r_x$ such that $\hat{u}(r_x) = u(A'(x)) - u(A'_{min})$
            \State $\mathcal{RG} \gets \mathcal{RG} \cup \{r_x\}$
        \EndFor
        \State Run the MWRR with input $\tree$, items $\mathcal{RG}$, and identical valuation $\hat{u}$, and obtain the multilevel allocation $\hat{\pi}$
        \State Initialise multilevel allocation $\pi$ such that $\pi(1) = \items$ and $\pi(i) = \{\}$ for any $i \in \nodes \setminus \{1\}$
        \For{$i \in \nodes \setminus \{1\}$}
            \For{$x \in \leaves{1}$}
                \If{$r_x \in \hat{\pi}(i)$}
                    \State $\pi(i) \gets \pi(i) \cup A'(x)$
                \EndIf
            \EndFor
        \EndFor
        \State \Return $\pi$
    \end{algorithmic}
    \end{algorithm}
    
    The proof the allocation is EFX at the leaves is exactly the same than in Theorem~3.5 of \cite{Scarlett_et_al2023}. Moreover, we proved in Theorem~3 that the MWRR was M-WEF1 under root-child-common valuations. All-common valuations being a special case of root-child-common valuations, we get the same guarantee. Recall that under identical preferences, the three adaptations of M-WEF1 coincide. Hence, we focus on showing that the returned allocation is M-WEF1. We extend the proof of \cite{Scarlett_et_al2023} to apply to the multilevel setting. 

    After the SM algorithm ran at the leaves, which computed an EFX allocation between the leaves, denote $A \in \allocations^\items_{\leaves}$ the resulting allocation at the leaves, and $\pi$ the induced multilevel allocation, i.e. $\pi(x) = A(x)$ for any $x \in \leaves$, and $\pi(i) = \cup_{x \in \leaves(i)} A(x)$ for any $i \in \internal$. Consider the set of bundles $\{A_1, \ldots, A_{\ell}\}$, where $\ell = |\leaves|$. We relabel it so that $u(A_1) \geq u(A_2) \geq \ldots \geq u(A_{\ell})$. Note that we denote the utilities of the leaves by $u$ as they are identical. For any leaf $x \in \leaves$, we define the representative good value to be $\hat{u}(r_x) = u(A_x) - u(A_\ell)$, where $r_x$ is the representative good of bundle $A_x$. We make the following two claims:
    \begin{claim}[1]
        For all $x \in \leaves$, $\hat{u}(r_x)$ is upper-bounded by the value of any good in $A_x$.
    \end{claim}
    \begin{claim}[2]
        Given any node $i \in \internal$, and any two of its children $j, k \in \children(i)$. Denote $\hat{\pi}$ the multilevel allocation of representative goods allocated to each node in $\nodes$, resulting from allocation $A$ to the leaves. If $\hat{\pi}$ is a M-WEF1 allocation of representative goods, then it remains M-WEF1 by replacing the representative goods by their corresponding bundles.
    \end{claim}
    Claim 1 holds because the allocation $A$, hence $\pi$, is EFX between the leaves, and therefore, for any leaf $x \in \leaves$, and any item $g \in A_x$, we have $u(A_x \setminus \{g\}) \leq u(A_\ell)$. Then, as valuations are additive, $u(A_x) - u(\{g\}) \leq u(A_\ell)$, and hence $\hat{u}(r_x) = u(A_x) - u(A_\ell) \leq u(\{g\})$. 
    
    We now prove Claim~2. Assume we have a M-WEF1 allocation of the representative goods, and denote it $\hat{\pi}$. For any two children $j, k$ of node $i \in \internal$, there must exist a representative good $r_{\max} \in \hat{\pi}(k)$ such that $\hat{u}(r_{\max}) = \max_{y \in \leaves(k)} \hat{u}(r_y)$, and the following holds: $$\frac{\sum_{x \in \leaves(j)} \hat{u}(r_x)}{w_j} \geq \frac{\sum_{y \in \leaves(k)} \hat{u}(r_y) - \hat{u}(r_{\max})}{w_k}$$

    By the definition of a representative good, for any leaf $x \in \leaves(j)$, we have $\hat{u}(r_x) = u(A_x) - u(A_\ell)$, and since $A_\ell$ is the least-valued bundle of A, we get 
    \begin{align*}
        \frac{\sum_{x \in \leaves(j)} \hat{u}(r_x)}{w_j} & = \frac{\sum_{x \in \leaves(j)} u(A_x) - u(A_\ell)}{w_j} \\
        & = \frac{\sum_{x \in \leaves(j)} u(A_x)}{w_j} - u(A_\ell), \quad \text{since $w_j = |\leaves(j)|$}
    \end{align*}

    Moreover, the right-hand side, for the same reason, can be rewritten as :
    \[
    \frac{\sum_{y \in \leaves(k)} \hat{u}(r_y) - \hat{u}(r_{\max})}{w_k}
    = 
    \frac{\sum_{y \in \leaves(k)} u(A_y) - \hat{u}(r_{\max})}{w_k}
    - u(A_\ell)
    \]

    Since in this setting, $u(\pi(j)) = \sum_{x \in \leaves(j)} u(A_x)$ and $u(\pi(k)) = \sum_{y \in \leaves(k)} u(A_y)$, we get 
    % $$\frac{v_j(\pi)}{w_j} = \frac{\sum_{x \in \leaves(j)} u(A_x)}{w_j} \geq \frac{\sum_{y \in \leaves(k)} u(A_y) - \hat{u}(r_{\max})}{w_k} = \frac{v_k(\pi) - \hat{u}(r_{\max})}{w_k}$$

    \begin{align*}
    \frac{u(\pi(j))}{w_j} = & \frac{\sum_{x \in \leaves(j)} u(A_x)}{w_j} \geq \\[4pt] 
    & \frac{\sum_{y \in \leaves(k)} u(A_y) - \hat{u}(r_{\max})}{w_k} = \frac{u(\pi(k)) - \hat{u}(r_{\max})}{w_k}.
\end{align*}

    Finally, from Claim~1, we know $\hat{u}(r_{\max}) \leq u(g_{\max})$, where $g_{\max}$ is some good from bundle $\pi(k)$ because by definition of a representative good, if $\hat{u}(r_{\max}) = \max_{y \in \leaves(k)} \hat{u}(r_y)$, then $u(g_{\max}) \geq u(g)$ for any $g \in \pi(k)$. Hence, we get $$\frac{u(\pi(j))}{w_j} \geq \frac{u(\pi(k)) - \hat{u}(r_{\max})}{w_k} \geq \frac{u(\pi(k)) - u(g_{\max})}{w_k}$$

    In SM-MWRR, we first compute an allocation to the leaves that is EFX between the leaves, then we construct the representative goods sets and values. Then, MWRR is run on the given tree, but with the representative goods to allocate and leaves equipped with the representative good utility function, $\hat{u}$. Since all leaves have the same utility function, we know from Theorem~3 that the returned allocation, denoted $\hat{\pi}$, is M-WEF1 w.r.t. $\hat{u}$. By Claim~2, we know this implies that the allocation w.r.t. the true bundles of items and utility functions $u$. This concludes the proof. 
\end{proof}

We then prove Corollary~1. 

\setcounter{corollary}{0}
\begin{corollary}
    Under root-child-common valuations and $w_i = |\leaves(i)|, \forall i \in \nodes$, the MWRR is M-WEF1 and EF1 between all leaves in $\leaves$.
\end{corollary}

\begin{proof}
    From Theorem~3, we know that MWRR satisfies M-WEF1. Furthermore, the proof of Theorem~3.9 in \cite{Scarlett_et_al2023} readily extends to our setting and establishes EF1 among the leaves. Their argument proceeds as follows. 
    
    Consider a single execution of the algorithm and partition the resulting sequence of allocations into $\lceil \frac{m}{|\leaves|} \rceil$ rounds, where in each round every leaf receives exactly one item. At any given round, a leaf $x$ prefers her bundle to that of any other leaf $y$ who picks after her in the same round. Moreover, leaf $x$ prefers the item she selects in a given round to the item selected in the next round by any leaf $y$ who picked before her in the current round. Consequently, at the end of the algorithm, leaf $x$ prefers her final bundle to that of any leaf $y$ appearing after her in the picking order, and prefers her bundle to that of any leaf $y$ appearing before her, up to the removal of the first item allocated to $y$.  
\end{proof}

\section*{B. Missing proofs from Section~4}

We show that $\overline{v}_i(S)$ is indeed a convex combination of the $(u_x)_{x \in \leaves(i)}$ for any internal node $i \in \internal$.

\setcounter{lemma}{2}
\begin{lemma}
    For any internal node $i \in \internal$ and any bundle $S \subseteq \items$, the agnostic estimated utility function $\overline{v}_i(S)$ is a convex combination of the $(u_x)_{x \in \leaves(i)}$.
\end{lemma}

\begin{proof}
    First, recall that we have $$\overline{v}_i(S) = \sum_{x \in \leaves(i)} \sum_{g \in S} u_x(g) \cdot W(x, i)$$

    We show that $\sum_{x \in \leaves(i)} W(x, i) = 1$.

    \begin{align*}
        \sum_{x \in \leaves(i)} W(x, i) & = \sum_{j \in \children(i)} W(j, i) \sum_{x \in \leaves(j)} W(x, j) \\
        & = \sum_{j \in \children(i)} W(j, i)\hspace{0.3cm} \ldots \sum_{k \in \children(\parent(k))} W(k, \parent(k)) \underbrace{\sum_{x \in \children(k)} W(x, k}_{=1}) \\
        & = \sum_{j \in \children(i)} W(j, i)\hspace{0.3cm} \ldots \underbrace{\sum_{k \in \children(\parent(k))} W(k, \parent(k)) \cdot 1}_{=1}) \\
        & = 1
    \end{align*} 
\end{proof}

We then provide the full proof of Theorem~4.

\setcounter{theorem}{3}
\begin{theorem}
    MWRR always returns a M[pess]-WEF1 allocation.
\end{theorem}

\begin{proof}

We show that MWRR computes an allocation $\pi \in \mallocations$ such that $\pi \vert_{\children(p)}$ is pessimistic WEF1 for any $p \in \internal$.

First, notice that for any $p$ such that $h(p) = 1$, the proof follows immediately from the proof of \cite{Chakraborty_et_al_2020} that their "least weight-adjusted frequent picker" procedure is WEF1. Indeed, at such node $p$ the children in $\children(p)$ are leaves and hence can compare directly their bundle with that on their sibling. Hence, MWRR acts exactly the same way than their algorithm. The only difference is in the tie-breaking scheme, which does not intervene in their proof.

Hence, our objective is to prove that the result holds for any $p \in \internal$ such that $h(p) \geq 2$. First, notice that the first $|\children(p)|$ picks are a simple round-robin, and each child receives one item. Hence, this first round is obviously pessimistic WEF1 as each child has at most one item at this point. Then, we focus on showing that the allocation computed remains WEF1 after the first pick.

\begin{lemma}[From \cite{Chakraborty_et_al_2020}]
    Consider an internal node $p \in \internal$ and one of its children $i \in \children(p)$ selected by MWRR at some iteration $t$, and suppose it is not $i$'s first pick. Let $t_i$ and $t_j$ be the numbers of times child $i$ and some other child $j \in \children(p)$ appear in the prefix of iteration $t$ (not including $t$ itself). Then, $\frac{t_j}{t_i} \geq \frac{w_j}{w_i}$.
\end{lemma}

\begin{proof}
    Since $i$ was picked at iteration $t$, it must be that $i \in \arg \min_{i' \in \children(p)} \frac{t_{i'}}{w_{i'}}$. Thus, for any $j \in \children(p)$, we have $\frac{t_i}{w_i} \leq \frac{t_j}{w_j}$. Hence, $\frac{t_j}{t_i} \geq \frac{w_j}{w_i}$.  
\end{proof}

\begin{lemma} \label{lemma: u>min_u}
    Consider MWRR at some iteration $t$ picked an internal node $i \in \internal$, and $x \in \leaves(i)$ was eventually selected to pick an item among the remaining items, denoted $\mathcal{RG}(t)$. Assume $x$ chooses $g \in \mathcal{RG}(t)$, then $$u_x(g) \geq \min_{y \in \leaves(i)} u_y(g'), \quad \forall g' \in \mathcal{RG}(t)$$
\end{lemma}

\begin{proof}
    Let $x \in \leaves(i)$ pick item $g \in \mathcal{RG}(t)$ among the remaining items at iteration $t$. Then, by construction, we have $$g \in \arg \max_{g' \in \mathcal{RG}(t)} u_x(g')$$

    Suppose by contradiction that $\exists g' \in \mathcal{RG}(t)$ such that $$u_x(g) < \min_{y \in \leaves(i)} u_y(g')$$

    Since $g \in \arg \max_{g' \in \mathcal{RG}(t)} u_x(g')$, it must be in particular that $$u_x(g') < \min_{y \in \leaves(i)} u_y(g')$$

    However since $x \in \leaves(i)$, it yields a contradiction. 
\end{proof}

We then proceed to show that, for any two children $i, j \in \children(p)$, MWRR computes an allocation such that $j$ is pessimistic weighted envy-free up to the first chosen item by $i$. To show this, we keep the same proof of \cite{Chakraborty_et_al_2020}, and  use our Lemma~\ref{lemma: u>min_u} at some point in their proof to be able to use the same argument adapted to our setting. 

\begin{lemma}[Mostly from \cite{Chakraborty_et_al_2020}] \label{lemma: pess wef1}
    Suppose that, for every iteration $t$ in which agent $i$ picks an item (i.e. one of its leaves eventually does) after her first pick, we have $t_i$ and $t_j$ for some other agent $j \in \children(p)$ satisfies $\frac{t_j}{t_i} \geq \frac{w_j}{w_i}$. Then, in the partial multilevel allocation up to and including $i$'s latest pick, agent $j$ is pessimistic weighted envy-free towards $i$ up to the first item picked. 
\end{lemma}

\begin{proof}
    Let $\gamma := \frac{w_j}{w_i}$. Consider any iteration $t$ in which agent $i$ is chosen after her first pick. Let agent $j$'s minimum values over its leaves in $\leaves(j)$ for the items allocated to agent $i$ in the latter's second, third, ..., $(t_i + 1)^\text{st}$ picks (the last one occuring at current iteration $t$ be $\beta_1, \beta_2, \ldots, \beta_{t_i}$ respectively. For instance, we have $\beta_1 = \min_{y \in \leaves(j)} u_y(g^2)$ where $g^2$ is the item selected at agent $i$'s second pick. Notice that $\beta$'s are different than those of \cite{Chakraborty_et_al_2020}.

    If $g^*$ is the first item picked by agent $i$ (i.e. by one if its leaves) and $\pi^t$ is the partial multilevel allocation up to and including iteration $t$, then clearly $\overset{\vee}{v}_j(\pi^t(i) \setminus \{g^*\}) = \sum_{x = 1}^{t_i} \beta_{t_i}$. 

    Let the number of times agent $j$ appears in the prefix of agent $i$'s second pick be $\tau_1$ ; that between agent $i$'s second and third picks be $\tau_2$ ; $\ldots$ ; that between agent $i$'s $t_i^{\text{th}}$ and $(t_i + 1)^{\text{th}}$ picks be $\tau_{t_i}$. 

    Let agent $j$'s values for the items she herself picked during phase $x \in [t_i]$ be $\alpha^x_1, \alpha^x_2, \ldots, \alpha^x_{\tau_x}$ (i.e. between agent $i$'s $x^\text{th}$ and $(x+1)^{\text{th}}$ picks, agent $j$'s received $\tau_x$ items). Notice that here when we say "agent $j$'s values for the items", we actually mean that for each item, we call $\alpha^x_k$ the utility of the leaf that picked the item at iteration $k$ during phase $x$. Then, we have $v_j(\pi^t) = \sum_{x = 1}^{t_i} \sum_{y = 1}^{\tau_x} \alpha^x_{y}$.

    Let, for $r \in [t_i]$, $\sum_{x=1}^{r} \tau_x$ and $r$ be the numbers of times agents $j$ and $i$ appear in the prefix of the latter's $(r+1)^{\text{th}}$ pick respectively. The condition of the lemma that $\frac{t_j}{t_i} \geq \frac{w_j}{w_i}$ yields 
    \begin{equation}
        \sum_{x=1}^{r} \tau_x \geq r \gamma
    \end{equation}

    Note that while $\tau_1 \geq \gamma > 0$ (because each agent is selected once in the first $\children(p)$ picks), it can be that $\tau_x = 0$ for $x \in \{2, 3, \ldots, t_i\}$. It corresponds to the scenario where agent $i$ picked more than once without agent $j$ picking in between.

    Moreover, we know from Lemma~\ref{lemma: u>min_u} that every time agent $j$ was chosen, the item eventually picked by the selected leaf in $\leaves(j)$ has a value greater or equal than the worst utility over all leaves in $\leaves(j)$ for any of the remaining items at this iteration, including those eventually picked by agent $i$ later. Hence, if $\tau_x > 0$ for some $x \in [t_i]$, we have $$\alpha^x_y \geq \max \{\beta_x, \beta_{x+1}, \ldots, \beta_{t_i}\}, \quad \forall y \in [\tau_x]$$
    Note that this is where Lemma~\ref{lemma: u>min_u} is necessary to carry on with the same proof than \cite{Chakraborty_et_al_2020}. Summing over all $y$'s, we get
    \begin{equation}
        \sum_{y = 1}^{\tau_x} \alpha^x_y \geq \tau_x \max \{\beta_x, \beta_{x+1}, \ldots, \beta_{\tau_x}\}
    \end{equation}
    Note that Inequality (2) holds trivially for $\tau_x = 0$ since both sides are zero. Hence, it holds for any $x \in [t_i]$.

    Now, we claim the following for each $r \in [t_i]$ $$\sum_{x=1}^{r} \sum_{y=1}^{\tau_x} \alpha^x_y \geq \gamma \sum_{x = 1}^{r} \beta_x + \big(\sum_{x=1}^{r} \tau_x - r\gamma  \big) \max \{\beta_r, \beta_{r+1}, \ldots, \beta_{t_i}\}$$

    To prove the claim, we proceed by induction on $r$. For the base case $r = 1$, we obtain from Inequality~(2)
    \begin{align*}
        \sum_{y=1}^{\tau_1} \alpha^1_y & \geq \tau_1 \max \{\beta_1, \beta_2, \ldots, \beta_{t_i}\} \\
        & \geq \gamma \beta_1 + (\tau_1 - \gamma) \max \{\beta_1, \beta_2, \ldots, \beta_{t_i}\}
    \end{align*}
    where second line follows from $\gamma \max \{\beta_1, \beta_2, \ldots, \beta_{t_i}\} \geq \gamma \beta_1$.

    For the inductive step, assume the claim holds for $r-1$. We now prove it for $r$. 
    \begin{align*}
        \sum_{x = 1}^{r} \sum_{y = 1}^{\tau_x} &\alpha^x_y = \sum_{x=1}^{r=1} \sum_{y=1}^{\tau_x} \alpha^x_y + \sum_{y=1}^{\tau_r} \alpha^r_y \\
        & \geq \gamma \sum_{x=1}^{r-1} \beta_x + \big(\sum_{x =1}^{r-1} \tau_x - (r-1)\gamma \big) \max \{\beta_{r-1}, \beta_r, \ldots, \beta_{t_i}\} + \sum_{y = 1}^{\tau_r} \alpha^r_y \\
        & \geq \gamma \sum_{x=1}^{r-1} \beta_x + \big(\sum_{x =1}^{r-1} \tau_x - (r-1)\gamma \big) \max \{\beta_{r-1}, \beta_r, \ldots, \beta_{t_i}\} + \tau_r \max \{\beta_r, \ldots, \beta_{t_i}\} \\
        & \geq \gamma \sum_{x=1}^{r-1} \beta_x + \big(\sum_{x =1}^{r-1} \tau_x - (r-1)\gamma \big) \max \{\beta_r, \ldots, \beta_{t_i}\} + \tau_r \max \{\beta_r, \ldots, \beta_{t_i}\} \\
        & = \gamma \sum_{x=1}^{r-1} \beta_x + \big(\sum_{x =1}^{r} \tau_x - (r-1)\gamma \big) \max \{\beta_r, \ldots, \beta_{t_i}\} \\
        & = \gamma \sum_{x=1}^{r-1} \beta_x + \gamma \max \{\beta_r, \ldots, \beta_{t_i}\} + \big(\sum_{x =1}^{r} \tau_x - r \gamma \big) \max \{\beta_r, \ldots, \beta_{t_i}\} \\
        & \geq \gamma \sum_{x=1}^{r-1} \beta_x + \gamma \beta_r + \big(\sum_{x =1}^{r} \tau_x - r \gamma \big) \max \{\beta_r, \ldots, \beta_{t_i}\} \\
        & = \gamma \sum_{x=1}^{r} \beta_x + \big(\sum_{x =1}^{r} \tau_x - r \gamma \big) \max \{\beta_r, \ldots, \beta_{t_i}\} \\
    \end{align*}
    where second line follows from the inductive hypothesis. The third line comes from Inequality~(2). The forth line follows from $\sum_{x = 1}^{r-1} \tau_x - (r-1)\gamma \geq 0$ from Inequality~(1) and $\max \{\beta_{r-1}, \beta_r, \ldots, \beta_{t_i}\} \geq \max \{\beta_r, \ldots, \beta_{t_i}\}$. 
    
    This completes the induction.

    Now, let $r = t_i$, we get 
    \begin{align*}
        \sum_{x = 1}^{t_i} \sum_{y=1}^{\tau_{t_x}} \alpha^x_y & \geq \gamma \sum_{x = 1}^{t_i} \beta_x + (\sum_{x = 1}^{t_i} \tau_x - t_i \gamma) \beta_{t_i} \\
        & \geq \gamma \sum_{x = 1}^{t_i} \beta_x
    \end{align*}
    where second line follows from Inequality~(1).

    Hence, as $v_j(\pi) = \sum_{x=1}^{t_i} \sum_{y = 1}^{\tau_x} \alpha^x_y$ and $\overset{\vee}{v}_j(\pi(i) \setminus \{g^*\}) = \sum_{x = 1}^{t_i} \beta_x$, we get
    \begin{align*}
        & v_j(\pi) \geq \gamma \cdot \overset{\vee}{v}_j(\pi(i) \setminus \{g^*\}) \\
        \Leftrightarrow & \frac{v_j(\pi)}{w_j} \geq \frac{\overset{\vee}{v}_j(\pi(i)) \setminus \{g^*\}}{w_i}
    \end{align*}
    Hence, agent $j$ is pessimistic weighted envy-free towards $i$ up to the first item agent $i$ picked. 
\end{proof}

Hence, since Lemma~\ref{lemma: pess wef1} holds for any $p \in \internal$ and any of its children, it concludes the proof. 
    
\end{proof}

\section*{C. Additional experimental results}

We provide hereafter further details on the experimental results we obtain. 

\vspace{0.1cm}
\noindent \textbf{Protocol.} Experiments were conducted on a server equipped with two Intel Xeon E5-2690v3 CPUs running at 2.60GHz, and 192 GB of RAM (each experiment was run on 4 cores). The program is written in Python. All results were obtained over 200 instances.

\vspace{0.1cm}
\noindent \textbf{Preference generation.} In order to design a robust experimental protocol, we use four different methods to generate the preferences of our leaves. 

\vspace{0.05cm}
\noindent \textit{Mallows.} We adapt the Mallows model \cite{Mallows1957}, commonly used in voting theory, to a multilevel setting. Although it generates ordinal preferences, it remains useful for fair allocation when we want similar rankings over items across agents. The model relies on a central ranking and a dispersion parameter $\phi \in [0,1]$: when $\phi = 0$, rankings match the central ranking exactly, while $\phi = 1$ yields uniformly random rankings (Impartial Culture).

In the multilevel version, each internal node $i$ is assigned its own dispersion parameter $\phi_i$. Starting from a uniformly sampled root ranking, we recursively generate central rankings for each node: the children of node $i$ receive rankings drawn from a Mallows model centered at $i$’s ranking with parameter $\phi_i$. This process continues down the tree until all leaves are assigned rankings, which we convert into (non-normalized) utilities using the \texttt{pref\_voting} library.

This framework enables fine control over preference correlations across the hierarchy: low $\phi_i$ enforces similarity among siblings, while higher values introduce greater diversity.

\vspace{0.05cm}
\noindent \textit{Resampling-Dirichlet.} The second method we use was introduced in \cite{Bohm_et_al2026}. Our multilevel adaptation works as follows: we first generate randomly a central approval vector $V$, where $m*p$ items are approved, with $p \in \{0.3, 0.6, 0.9\}$ a probability. Then, starting at the root, we generate approval vector for each of its children -- if item $g$ is approved in $V$, then child $i$ approves it with probability $p$, and if $g$ is not approved, then $i$ approves it with probability $\phi \in \{0.2, 0.8\}$. We repeat this procedure until reaching the leaves. 

For each leaf $x$, we construct a utility vector $u_x$. Let $App$ be the ordered list of items approved by $x$. For a parameter $t \geq 0$, we assign utilities $u_x(g) = 10^{-10}$ for non-approved items, and $u_x(g) = \frac{2}{(j/|App| + 0.01)^t}$ for the $j$-th approved item. If no item is approved, one is selected uniformly at random. Finally, we sample a normalized utility vector from a Dirichlet distribution parametrized by this vector.
 
\vspace{0.05cm}
\noindent \textit{Cost utilities.} The third method, proposed in \cite{Botan_et_al2023}, assumes a utility vector describing the public value of each item (drawn uniformly at random). Each agent can either approve or disapprove an item. It she approves it, then her utility for the item is the public utility ; otherwise, she has no utility for it. We adapt it to the multilevel setting by computing from the root to the leaves, an approval vector for each node $i$ as follows: for any item $g \in \items$, if $\parent(i)$ approves it, then $i$ approves with probability $p \in \{0.3, 0.6, 0.9\}$, while if $\parent(i)$ disapproves it, then $i$ approves it with probability $\phi \in \{0.2, 0.8\}$.    

\vspace{0.05cm}
\noindent \textit{Correlated utilities.} The last method, proposed in \cite{Dickerson_et_al2014}, assumes again a central utility vector $V$, drawn uniformly at random. The utility of any agent $x$ for any item $g$ is then drawn from a normal distribution whose mean is $V(g)$, i.e. $u_x(g) \sim \mathcal{N}(V(g), \sigma_x)$ with $\sigma_x \in \{0.2, 0.8\}$ the standard deviation associated with agent $x$. Our adaptation to the multilevel setting follows the same logic: we draw uniformly at random a central vector. Starting at the root, we sample some utility vector for the children of the root, i.e. for each $i \in \children(1)$ and each item $g \in \items$, $u_i(g) \sim \mathcal{N}(V(g), \sigma_1)$. We repeat this process until reaching the leaves. 

\vspace{0.1cm}
\noindent \textbf{Results.} In the main body of the paper, we reported our experimental results and observed that, among the 192,000 generated instances, fewer than 50 led MWRR to return a non-M[agno]-WEF1 allocation. In this section, we provide the detailed tables corresponding to those instances. We observe that, across all preference generation methods, non-M[agno]-WEF1 allocations occur only for instances defined on comb trees with random weights (See Tables~\ref{tab:fair1_rwtrue_bal0_mallows}, \ref{tab:fair1_rwtrue_bal0_resampling}, \ref{tab:fair1_rwtrue_bal0_cost}, and \ref{tab:fair1_rwtrue_bal0_correlated}). Note, however, that such unfair allocations are not limited to this class of instances; they also arise in other settings (see, for example, the proof of Theorem~6).

% ── Mallows - RW=True  BALANCED=0  metric=fair1 ────────────────────
\begin{table}[htbp]
    \centering
    \caption{\centering Proportion of non-M[agno]-WEF1 allocations (95\% CI) - Mallows - RW - Comb tree. Columns: $(\varphi_1, \varphi_2)$.}
    \label{tab:fair1_rwtrue_bal0_mallows}
    \begin{tabular}{llcccc}
        \toprule
        $N$ & $M$ & $(2,2)$ & $(2,8)$ & $(8,2)$ & $(8,8)$ \\
        \midrule
        15 & 15 & $[0.00,\, 0.01]$ & $[0.00,\, 0.00]$ & $[0.00,\, 0.00]$ & $[0.00,\, 0.00]$ \\
        15 & 30 & $[0.00,\, 0.00]$ & $[0.00,\, 0.00]$ & $[0.00,\, 0.00]$ & $[0.00,\, 0.00]$ \\
        31 & 31 & $[0.00,\, 0.00]$ & $[0.00,\, 0.00]$ & $[0.00,\, 0.00]$ & $[0.00,\, 0.00]$ \\
        31 & 62 & $[0.00,\, 0.00]$ & $[0.00,\, 0.00]$ & $[0.00,\, 0.00]$ & $[0.00,\, 0.01]$ \\
        63 & 63 & $[0.00,\, 0.00]$ & $[0.00,\, 0.00]$ & $[0.00,\, 0.00]$ & $[0.00,\, 0.00]$ \\
        63 & 126 & $[0.00,\, 0.00]$ & $[0.00,\, 0.00]$ & $[0.00,\, 0.00]$ & $[0.00,\, 0.00]$ \\
        127 & 127 & $[0.00,\, 0.00]$ & $[0.00,\, 0.00]$ & $[0.00,\, 0.00]$ & $[0.00,\, 0.00]$ \\
        127 & 254 & $[0.00,\, 0.00]$ & $[0.00,\, 0.00]$ & $[0.00,\, 0.00]$ & $[0.00,\, 0.00]$ \\
        \bottomrule
    \end{tabular}
\end{table}

% ── Resampling-Dirichlet RW=True  BALANCED=0  metric=fair1 ────────────────────
\begin{table}[htbp]
    \centering
    \caption{\centering Proportion of non-M[agno]-WEF1 allocations (95\% CI) - resampling dirichlet - RW - Comb tree. Columns: $(p, \varphi)$.}
    \label{tab:fair1_rwtrue_bal0_resampling}
    \begin{tabular}{llcccccc}
        \toprule
        $N$ & $M$ & $(3,2)$ & $(3,8)$ & $(6,2)$ & $(6,8)$ & $(9,2)$ & $(9,8)$ \\
        \midrule
        15 & 15 & $[0.00,\, 0.05]$ & $[0.00,\, 0.01]$ & $[0.00,\, 0.02]$ & $[0.00,\, 0.01]$ & $[0.00,\, 0.00]$ & $[0.00,\, 0.00]$ \\
        15 & 30 & $[0.00,\, 0.02]$ & $[0.00,\, 0.00]$ & $[0.00,\, 0.00]$ & $[0.00,\, 0.00]$ & $[0.00,\, 0.00]$ & $[0.00,\, 0.00]$ \\
        31 & 31 & $[0.00,\, 0.02]$ & $[0.00,\, 0.01]$ & $[0.00,\, 0.00]$ & $[0.00,\, 0.01]$ & $[0.00,\, 0.00]$ & $[0.00,\, 0.00]$ \\
        31 & 62 & $[0.00,\, 0.00]$ & $[0.00,\, 0.00]$ & $[0.00,\, 0.00]$ & $[0.00,\, 0.00]$ & $[0.00,\, 0.00]$ & $[0.00,\, 0.00]$ \\
        63 & 63 & $[0.00,\, 0.00]$ & $[0.00,\, 0.00]$ & $[0.00,\, 0.00]$ & $[0.00,\, 0.01]$ & $[0.00,\, 0.00]$ & $[0.00,\, 0.00]$ \\
        63 & 126 & $[0.00,\, 0.00]$ & $[0.00,\, 0.00]$ & $[0.00,\, 0.00]$ & $[0.00,\, 0.00]$ & $[0.00,\, 0.00]$ & $[0.00,\, 0.00]$ \\
        127 & 127 & $[0.00,\, 0.00]$ & $[0.00,\, 0.00]$ & $[0.00,\, 0.00]$ & $[0.00,\, 0.00]$ & $[0.00,\, 0.00]$ & $[0.00,\, 0.00]$ \\
        127 & 254 & $[0.00,\, 0.00]$ & $[0.00,\, 0.00]$ & $[0.00,\, 0.00]$ & $[0.00,\, 0.00]$ & $[0.00,\, 0.00]$ & $[0.00,\, 0.00]$ \\
        \bottomrule
    \end{tabular}
\end{table}

% ── Cost-utilities - RW=True  BALANCED=0  metric=fair1 ────────────────────
\begin{table}[htbp]
    \centering
    \caption{\centering Proportion of non-M[agno]-WEF1 allocations (95\% CI) - Cost utilities - RW - Comb tree. Columns: $(p, \varphi)$.}
    \label{tab:fair1_rwtrue_bal0_cost}
    \begin{tabular}{llcccccc}
        \toprule
        $N$ & $M$ & $(3,2)$ & $(3,8)$ & $(6,2)$ & $(6,8)$ & $(9,2)$ & $(9,8)$ \\
        \midrule
        15 & 15 & $[0.00,\, 0.04]$ & $[0.00,\, 0.01]$ & $[0.00,\, 0.01]$ & $[0.00,\, 0.01]$ & $[0.00,\, 0.00]$ & $[0.00,\, 0.00]$ \\
        15 & 30 & $[0.00,\, 0.03]$ & $[0.00,\, 0.01]$ & $[0.00,\, 0.01]$ & $[0.00,\, 0.00]$ & $[0.00,\, 0.00]$ & $[0.00,\, 0.00]$ \\
        31 & 31 & $[0.00,\, 0.04]$ & $[0.00,\, 0.01]$ & $[0.00,\, 0.01]$ & $[0.00,\, 0.00]$ & $[0.00,\, 0.00]$ & $[0.00,\, 0.00]$ \\
        31 & 62 & $[0.00,\, 0.01]$ & $[0.00,\, 0.00]$ & $[0.00,\, 0.00]$ & $[0.00,\, 0.00]$ & $[0.00,\, 0.00]$ & $[0.00,\, 0.00]$ \\
        63 & 63 & $[0.00,\, 0.00]$ & $[0.00,\, 0.00]$ & $[0.00,\, 0.00]$ & $[0.00,\, 0.00]$ & $[0.00,\, 0.00]$ & $[0.00,\, 0.00]$ \\
        63 & 126 & $[0.00,\, 0.00]$ & $[0.00,\, 0.00]$ & $[0.00,\, 0.00]$ & $[0.00,\, 0.00]$ & $[0.00,\, 0.00]$ & $[0.00,\, 0.00]$ \\
        127 & 127 & $[0.00,\, 0.00]$ & $[0.00,\, 0.00]$ & $[0.00,\, 0.00]$ & $[0.00,\, 0.00]$ & $[0.00,\, 0.00]$ & $[0.00,\, 0.00]$ \\
        127 & 254 & $[0.00,\, 0.00]$ & $[0.00,\, 0.00]$ & $[0.00,\, 0.00]$ & $[0.00,\, 0.00]$ & $[0.00,\, 0.00]$ & $[0.00,\, 0.00]$ \\
        \bottomrule
    \end{tabular}
\end{table}

% ── Correlated preferences - RW=True  BALANCED=0  metric=fair1 ────────────────────
\begin{table}[htbp]
    \centering
    \caption{\centering Proportion of non-M[agno]-WEF1 allocations (95\% CI) - Correlated preferences - RW - Comb tree. Columns: $(\varphi_1, \varphi_2)$.}
    \label{tab:fair1_rwtrue_bal0_correlated}
    \begin{tabular}{llcccc}
        \toprule
        $N$ & $M$ & $(2,2)$ & $(2,8)$ & $(8,2)$ & $(8,8)$ \\
        \midrule
        15 & 15 & $[0.00,\, 0.00]$ & $[0.00,\, 0.00]$ & $[0.00,\, 0.00]$ & $[0.00,\, 0.00]$ \\
        15 & 30 & $[0.00,\, 0.00]$ & $[0.00,\, 0.00]$ & $[0.00,\, 0.00]$ & $[0.00,\, 0.00]$ \\
        31 & 31 & $[0.00,\, 0.00]$ & $[0.00,\, 0.00]$ & $[0.00,\, 0.00]$ & $[0.00,\, 0.00]$ \\
        31 & 62 & $[0.00,\, 0.00]$ & $[0.00,\, 0.00]$ & $[0.00,\, 0.00]$ & $[0.00,\, 0.00]$ \\
        63 & 63 & $[0.00,\, 0.00]$ & $[0.00,\, 0.01]$ & $[0.00,\, 0.00]$ & $[0.00,\, 0.00]$ \\
        63 & 126 & $[0.00,\, 0.00]$ & $[0.00,\, 0.00]$ & $[0.00,\, 0.00]$ & $[0.00,\, 0.00]$ \\
        127 & 127 & $[0.00,\, 0.00]$ & $[0.00,\, 0.00]$ & $[0.00,\, 0.00]$ & $[0.00,\, 0.00]$ \\
        127 & 254 & $[0.00,\, 0.00]$ & $[0.00,\, 0.00]$ & $[0.00,\, 0.00]$ & $[0.00,\, 0.00]$ \\
        \bottomrule
    \end{tabular}
\end{table}

\end{document}